\documentclass[12pt]{article}  
\usepackage{makeidx}
\usepackage{amsfonts}
\usepackage{amsmath}
\usepackage{amssymb}
\usepackage{graphicx}
\usepackage[T1]{fontenc}
\usepackage{color}
\usepackage{enumerate}

\newtheorem{theorem}{Theorem}
\newtheorem{result}{Result}

\newtheorem{assumption}[theorem]{Assumption}

\newtheorem{condition}[theorem]{Condition}

\newtheorem{adefinition}[theorem]{Definition}

\newtheorem{aexample}[theorem]{Example}

\newtheorem{lemma}[theorem]{Lemma}

\newtheorem{aproblem}[theorem]{Problem}

\newtheorem{acomment}[theorem]{Comment}

\newtheorem{aremark}[theorem]{Remark}
\newenvironment{remark}{\begin{aremark}\rm}{\end{aremark}}

\numberwithin{equation}{section} \numberwithin{theorem}{section}
\newenvironment{proof}[1][Proof]{\textbf{#1.} }{\ \rule{0.5em}{0.5em}}

\DeclareMathOperator{\Var}{{\bf Var}}
\DeclareMathOperator{\tr}{Tr}
\DeclareMathOperator{\dist}{dist}
\DeclareMathOperator{\supp}{supp}
\newcommand\e{\mathrm{e}}
\newcommand\E{\mathbf E}
\newcommand\Z{\mathbb Z}
\newcommand\C{\mathbb C}
\newcommand\N{\mathbb N}

\newcommand\eps{\varepsilon}
\newcommand{\trip}{\|\kern-.08em |}
\newcommand\beq{\begin{equation}}
\newcommand\eeq{\end{equation}}
\newcommand{\abs}[1]{\left\lvert #1 \right\rvert}
\newcommand{\set}[1]{\left\{ #1 \right\}}
\newcommand{\pa}[1]{\left( #1 \right)}
\newcommand\La{\Lambda}

\begin{document}

\title{Large Block Properties \\
of the Entanglement Entropy \\
of Free Disordered Fermions}
\author{A.Elgart \\
Department of Mathematics, Virginia Tech, \\Blacksburg, VA,
24061, USA\\
L.Pastur, M. Shcherbina\\
Mathematical Division, B.Verkin Institute \\ for Low Temperature
Physics and Engineering \\  Kharkiv, 61103, Ukraine}
\maketitle
\date

\begin{abstract}
We consider a macroscopic disordered system of free $d$-dimensional lattice fermions whose
one-body Hamiltonian is a Schr\"{o}dinger operator $H$ with ergodic
potential. We assume that the Fermi energy lies in the exponentially localized
part of the spectrum of $H$.
We prove that if $S_\Lambda$ is the entanglement entropy of a
lattice cube $\Lambda$ of side length $L$ of the system, then for any $d
\ge 1$ the expectation $\mathbf{ E}\{L^{-(d-1)}S_\Lambda\}$ has a finite limit as $L \to
\infty$ and we identify the  limit. Next, we prove that for $d=1$ the entanglement
entropy admits a well defined asymptotic form for all typical realizations (with
probability 1) as $ L \to \infty$. According to numerical results of [33]
the limit is not selfaveraging even for an i.i.d. potential. On the other
hand, we show that for $d \ge 2$ and an
i.i.d. random potential the variance of $L^{-(d-1)}S_\Lambda$
decays polynomially  as $L \to \infty$, i.e., the entanglement entropy is
selfaveraging.
\end{abstract}


\section{Introduction}

\label{s0}

Entanglement is a fundamental feature of quantum mechanics  manifested in
non-local, intrinsically quantum correlations between separated quantum
systems. Used first  by Einstein, Rosen and Podolsky  in 1935  to
demonstrate the incompleteness of quantum description, entanglement was  coined and explicitly defined by Schr\"{o}dinger shortly thereafter.  Nowadays  entanglement is  an
object of extensive studies ranging from general relativity, cosmology, and
foundation of quantum mechanics through quantum optics and quantum
statistical mechanics to quantum information and computation. Among the wide
variety of ideas, problems and results  concerning entanglement phenomena, there ia a considerable amount of those dealing with many body (macroscopic) systems, common in statistical mechanics and
condensed matter physics  (see the recent reviews \cite{Am-Co:08,Ca-Co:11,Ei-Co:11,La:15}). Consider a bipartite macroscopic quantum system $\mathfrak{S}$
\begin{equation}\label{SBE}
\mathfrak{S}=\mathfrak{B} \cup \mathfrak{E},
\end{equation}
consisting of a block $\mathfrak{B}$ and its "environment"
$\mathfrak{E}$. It is assumed that $\mathfrak{S}$ occupies a macroscopic domain
$\mathcal{D} \subset \mathbb{Z}^d$ of
characteristic size $N$,  $\mathfrak{B}$ occupies a subdomain
$\Lambda \subset \mathcal{D}$  of the characteristic size $L$,
and one is interested in the degree to which $\mathfrak{B}$ and $\mathfrak{E}$ are correlated in the asymptotic regime
\begin{equation}
1\ll L\ll N.  \label{bip}
\end{equation}
For a pure state $\rho$ of $\mathfrak{S}$,  a widely used measure of the corresponding correlations is
the von Neumann entropy, defined by
\begin{equation}
S_{\Lambda}=-\tr_{\Lambda}\,\rho _{\Lambda}\log _{2}\rho
_{\Lambda},  \label{ente}
\end{equation}%
where $\rho_\Lambda$ denotes the reduced density matrix associated with the block $\mathfrak{B}$.
One of the central problems  in the field is the determination of the  asymptotic behavior of the entanglement entropy
in the asymptotic regime \eqref{bip}. One usually  takes the macroscopic limit $N \rightarrow \infty
$ for $\mathfrak{S}$ first whenever it is possible, which reduces the problem  to finding the large $L$ asymptotics  of $S_{\Lambda}$ for a block of size $L$ of the infinite many body system.

It has been found in the recent decades that the large block asymptotics of the entanglement entropy \eqref{ente} may be unusual if $\rho$ is a ground state of the system, or, more generally, an eigenstate of the system. Namely, it was shown in several physics works that  the entanglement entropy can be asymptotically proportional to  the surface area $L^{d-1}$ of the block rather than  its
volume $L^{d}$  as $L \to \infty$. The latter (extensive) asymptotics is standard for thermal states in quantum
statistical mechanics  \cite{Ru:77}, while the former was found
first in cosmology and quantum field theory and later in other fields, and is
known as \textit{area law} \cite{Am-Co:08,Ca-Co:11,Ei-Co:11}. It
has also been found that the area law asymptotics is not always valid,
e.g., at quantum critical points of several one-dimensional translation
invariant quantum spin chains, for which the  entanglement entropy grows like $\log L$   rather than remaining bounded \cite{Ca-Co:11}.

More generally,  area law asymptotics $S_{\Lambda}\sim L^{d-1}$ for
the entanglement entropy are believed to be valid
for quantum systems with finite range interaction and a gap between the ground
state energy and the rest of the spectrum, while other asymptotics are possible for gapless systems. In particular, some systems that have a quantum phase transition may exhibit asymptotics of the form $S_\Lambda \sim L^{d-1}\log L$, \cite{Am-Co:08,La:15}. Determining whether an eigenstate of the system is spectrally isolated from the rest of the spectrum is  generally a daunting task that was undertaken mostly for certain one dimensional exactly solvable models. On the other
hand, there is a class of simpler models that can be either gapless or gapped and exhibit accordingly either type of aforementioned asymptotics, in any dimension.

Concretely, examples of quantum systems with this property are given  by 
quasifree fermions described by Hamiltonians that  are quadratic in the creation and
annihilation operators. Such systems arise in  condensed matter theory
and statistical mechanics (models describing electrons in metals, including superconductors, other models with mean field type approximations, exactly
solvable spin chains, etc.).

For these Hamiltonians with finite range and translation
invariant coefficients the large $L$ behavior of the entanglement entropy in a gapless case was initially studied in \cite{Gi-Kl:06,He-Co:09,Wo:06}, where either the upper and lower bounds of order $O(L^{d-1} \log L)$ for the entanglement entropy were obtained or the asymptotic formula of the same order of magnitude was proposed by using certain conjectures on the subleading term in the Szeg\H{o} theorem for T\"{o}plitz determinants with discontinuous symbols.
The precise asymptotic behavior for such systems was recently established rigorously in
\cite{Le-Co:13} by using rather sophisticated techniques of modern
operator theory \cite{So:10,So:13,So:15}.

All results mentioned above deal with  translation invariant systems. Following a widely
accepted paradigm in condensed matter physics, it is natural to consider a
disordered version of the free fermion model. To this end, one  replaces the translation
invariant coefficients of the fermionic quadratic Hamiltonian by random coefficients,
which are translation invariant in the mean and have
decaying statistical correlation, i.e.,\ ergodic. This is the  standard setting for the theory of disordered  systems \cite{LGP,Pa-Fi:78}.

The analysis of the many body quadratic Hamiltonian reduces to that of a certain one
body operator determined by the coefficients of the original Hamiltonian
\cite{Bo-Bo:96}. Thus, in the case
of random coefficients we obtain a problem in the  theory of one body
disordered systems, related to the phenomenon of Anderson localization.

Specifically, one can consider the case, where $\mathfrak{S}$ of \eqref{SBE} is the system of free fermions in their ground- (or just eigen-) state having a discrete Schr\"{o}dinger operator $H$
 with random (more generally
ergodic, see (\ref{DS}) and (\ref{erpot}) -- (\ref{verg})) potential as the  one body operator. It is known that the spectrum of $H$ is
non random and consists of  intervals $[E_{2j-1},
E_{2j}]$, for $ j=1,...,p$, referred to as bands. Moreover, in certain adjacent to band edges subintervals  $[E_{2j-1},E_{2j-1}^{\prime}]$ and $[E_{2j}^{\prime},E_{2j}]$, for $E_{2j-1}^{\prime} \le E_{2j}^{\prime}$ and some $j$'s,  the spectrum   (especially in the case of i.i.d. potentials)  is almost surely of pure point type and the corresponding eigenfunctions
are exponentially localized.
We will call these subintervals \emph{exponentially localized parts of the
spectrum} of $H$. The parts can be characterized by the exponential decay of the expectations of the off diagonal entries of various important spectral characteristics, e.g. the spectral projections of $H$ (see (\ref{b_P}) and (\ref{PEI})), fractional powers of the Green function of $H$ (see (\ref{dynloc}) -- (\ref{dynloc})), etc. \cite{Ai-Wa:15,Pa-Fi:92}.

It was shown rigorously in \cite{Pa-Sl:14} (see also related works \cite%
{Na-Co:15,St-Ab:15}) that if the Fermi energy $\mu$ lies in either the spectral
gap or the exponentially localized part of the
spectrum of $H$, then
the expectation $\mathbf{E}\{S_{\Lambda}\}$ of the
entanglement entropy $S_{\Lambda}$ of a lattice cube
$%
\Lambda $ of side length $L$ admits a two-sided bound of the form
\begin{equation}
C_{-}L^{d-1}\leq \mathbf{E}\{S_{\Lambda}\}\leq C_{+}L^{d-1},\;0<C_{-}\leq
C_{+}<\infty.  \label{scal}
\end{equation}%
The spectral gap case is fairly simple and follows from  general
principles of spectral theory, while the gapless case follows from the exponential decay of the expectation of the off diagonal matrix elements of the Fermi projection, one of fundamental results in the theory of localization.
For $d=1$ and $L\gg 1$,  the two-sided bound for the entanglement entropy for almost all realizations of disorder was also obtained in \cite{Pa-Sl:14}
and then was used to show numerically that the entanglement entropy of
one dimensional disordered lattice fermions is not selfaveraging, i.e., has non
vanishing random fluctuations even if $L\gg 1$.

In this paper we will assume that the Fermi energy $\mu$ lies in the
exponentially localized parts of the spectrum.
We first prove that in any dimension  there exists a "surface
macroscopic" limit \begin{equation}
\lim_{L\rightarrow \infty }L^{-(d-1)}\mathbf{E}\{S_{\Lambda}\}  \label{alo}
\end{equation}%
of the entanglement entropy per unit surface area of a cubic block
$%
\Lambda$ with a side length $L$. The limit  is  not zero and finite in view of  \eqref{scal}, see Result \ref{t:2} and Theorem
\ref{thm:splitinf} below.

In other
words, the entanglement entropy of disordered fermions satisfies
\emph{area law in the mean}.

We then show that for $d \ge 2$ the variance of $L^{-(d-1)}S_{\Lambda}$
vanishes polynomially fast in $L$
as $L \to \infty$, i.e.,\ that for $d \ge 2$ the
entanglement entropy of disordered lattice fermions is selfaveraging, see
Result \ref{t:var} and Theorem \ref{thm:varia}.

For $d=1$  we establish that $S_{\Lambda}$ has a well defined
asymptotic form as $L\rightarrow \infty $ for all typical realizations of
disorder (with probability $1)$, see Result \ref{t:1p1} and Theorem \ref{t:fp1} below.
 According to the numerical results of \cite{Pa-Sl:14}, the corresponding asymptotic expression  is a non trivial random variable, i.e., the entanglement entropy of disordered lattice fermions is not selfaveraging in the one dimensional case.

Note that the selfaveraging property, i.e., the disappearance of fluctuations of appropriately normalized extensive observables  in the macroscopic limit, is widely known in condensed matter theory and statistical mechanics of disordered systems \cite{LGP,Pa-Fi:78}. In entanglement studies the essentially analogous property is known as entanglement typicality (see e.g.\ the recent reviews \cite{Da-Co:14,La:15}). Entanglement typicality allows one to consider the entanglement characteristics, which are "typical", i.e., random with respect to a certain multivariate probability distribution, provided that the distribution is strongly peaked in a number of variables.

It is worth mentioning that in  a number of studied cases the multivariate probability distribution is chosen to be the normalized Haar measure of the multidimensional unitary group $U(N)$,
which is unfortunately not always easy to interpret physically. In particular, it is not simple to identify unambiguously the
physical dimension and size of the quantum system
in question.  Note that both quantities enter explicitly in the large
block asymptotics of the entanglement entropy (cf. \eqref{scal} and
\eqref{alo}).
On the other hand, the ground state  (or an eigenstate, more generally) of $N$
free disordered lattice fermions is just the Slater determinant of $N$ eigenfunctions
of a $d$-dimensional Schr\"{o}dinger operator with a random (or in greater generality ergodic) potential and $\Lambda$ is just a cubic block in $\mathbb{Z}^d$.
In this simple framework, we can explicitly study the various entanglement properties and characteristics of the free fermion system, establishing,  in particular,  that the entanglement entropy per unit surface is typical for $d \ge 2$ and is not typical for $d=1$.

The paper is organized as follows. In  Section 2, we outline the
setting and formulate our main results for the large block behavior of the
entanglement entropy of free disordered fermions whose one body operator
is a Schr\"{o}dinger operator with an ergodic potential. These
results are particular cases of assertions, valid for more
general quantities and a broader class of one body ergodic operators.  We formulate and prove these assertions in Section 3. The proofs rely  on a number of auxiliary facts, which are in turn proved in
Section 4. In Section 5 we present  our outlook and draw our conclusions.

Throughout the paper we will use the symbols $C,C_{1},c,c_{1},$ etc. for quantities
which may be different in
different expressions and whose value is not essential for the validity of
the corresponding formulas. For a set $\mathcal{C} \subset \mathbb{Z}^d$, we will denote by
$\mathcal{C}^c=\mathbb{Z}^d \setminus \mathcal{C} $ its complement, by
$|\mathcal{C}|$ its cardinality and by $\chi_{\mathcal{C}}$ its indicator function (that can be thought of as the   projection operator from $\ell^2(\mathbb{Z}^d)$ onto $\ell^2(\mathcal{C})$).

\section{Results}

In this section we present our main results on the large block behavior of the entanglement entropy of free disordered fermions whose one body operator
is a Schr\"{o}dinger operator with an ergodic potential (see Results \ref{t:expvan} -- \ref{t:var}). The results are corollaries of more general facts (see Theorems \ref{thm:apriori} -- \ref{thm:varia}), which are formulated and proved in Sections 3 and 4.

\label{s:1}

\subsection{Generalities}

Let $A$ be a bounded hermitian operator acting on $\ell^{2}(\mathbb{Z}^{d})$ and let $A(x,y)=\langle\delta_x,A\delta_y\rangle$  be its $(x,y)\in\Z^d\times\Z^d$ matrix elements. We consider a system of spinless lattice fermions confined to a finite
domain $\mathcal{D}\subset \mathbb{Z}^{d}$ and described by the
Hamiltonian
\begin{equation}
\mathcal{H_{\mathcal{D}}}=\sum_{x,y\in \mathcal{D}}A(x,y)c_{x}^{+}c_{y},
\label{HF}
\end{equation}%
quadratic in  the Fermi creation and
annihilation operators $c_{x}^{+},c_{x},\;x\in \mathcal{D}$. The prototypical example for $A$ that we will focus on in this paper is
\begin{equation}
A=H-\mu ,  \label{AHm}
\end{equation}%
where  $\mu $ is a parameter (the Fermi energy) to be chosen below and $H$ is a discrete Schr\"{o}dinger operator
\begin{equation}
H=-\Delta +V,  \label{DS}
\end{equation}%
acting on $\ell^{2}(\mathbb{Z}^{d})$. Here, $\Delta $ is the $d$-dimensional discrete
Laplacian%
\begin{equation}
(\Delta \psi )(x)=\sum_{|x-y|=1}\psi (y),\;x\in \mathbb{Z}^{d},  \label{dl}
\end{equation}%
and%
\begin{equation}
(V\psi )(x)=V(x)\psi (x),\;x\in \mathbb{Z}^{d}.  \label{pot}
\end{equation}%
is the potential.

Note that the definition \eqref{ente} of the entanglement entropy for bipartite fermionic systems (or, more generally, indistinguishable particles) does not exactly coincide with that for quantum systems of distinguishable quantum
entities (spins, qubits), which are often considered in  quantum information
theory and related studies of quantum spin chains \cite%
{Am-Co:08,Ca-Co:11,Ei-Co:11}. Indeed, according to \eqref{ente}, the entanglement
entropy is determined by the reduced density matrix $\rho _{\Lambda }$. In the distinguishable case, one usually represents the Hilbert state space of
the bipartite system    $\mathfrak{S}$ \eqref{SBE} as the tensor product $\mathcal{H}_{\mathfrak{S}}=%
\mathcal{H}_{\mathfrak{B}}\otimes \mathcal{H}_{\mathfrak{E}}$ of the Hilbert
state spaces of parties $\mathfrak{B}$ and $\mathfrak{E}$. This allows one to introduce the partial traces $\mathrm{Tr}_{\Lambda }$, $\mathrm{Tr}_{\mathcal{D}\setminus \Lambda }$,  $\mathrm{Tr}_{\mathcal{D}}:=\mathrm{Tr}_{\mathcal{D}\setminus \Lambda }\mathrm{%
Tr}_{\Lambda }$, and define $\rho _{\Lambda }=\mathrm{Tr}_{\mathcal{D}%
\setminus \Lambda }\; \rho _{\mathcal{D}}$. Clearly, we cannot proceed in the same fashion in the indistinguishable case. Instead we will use the
definition of the reduced density matrix which is common in quantum
statistical mechanics and identifies it as the corresponding quantum
correlation function. Namely, let $O_{\Lambda }$ be the (local) subalgebra
of the algebra $O_{\mathcal{D}}$ of observables of the whole bipartite
composite \eqref{SBE}, i.e.,\ $O_{\Lambda }$ is the set of all polynomials $%
\{\pi _{\Lambda }\}$ in the creation and annihilation Fermi operators
indexed by the points in the set $\Lambda$ occupied by the block $\mathfrak{B}$. We then define $\rho
_{\Lambda }$ via the relation $\mathrm{Tr\;}\pi _{\Lambda }\rho _{\mathcal{D}}=
\mathrm{Tr\;}\pi _{\Lambda }\rho _{\Lambda }$ required to hold for all $%
\pi _{\Lambda }\in O_{\Lambda }$. We refer the reader to the works \cite%
{Am-Co:08,Be-Co:14,Ei-Co:11,Ki-Co:14,Sh:04} for discussions of this definition and
recent  results. Let us note that in the commutative case
(i.e.,\ in the classical statistical mechanics and probability theory)  both
definitions coincide. In addition,  entanglement entropy for both definitions
possesses the important property $S_{\Lambda }=S_{D\setminus \Lambda }$. For
the first definition it is a direct consequence of the Schmidt
decomposition of hermitian matrices \cite{Am-Co:08}. For the second 
(algebraic) definition in the free fermion 
case considered in this paper, this property follows from
formulas \eqref{pic}, \eqref{sh0} and \eqref{Pires} below.

 Performing a fairly standard second quantization computation (see e.g.
\cite{St-Ab:15,Bo-Bo:96,Pe-In:09} for details and references),  one verifies that the entanglement entropy,  corresponding to the block $%
\Lambda $ of the free fermion system  that occupies a domain
$%
\mathcal{D} \subset \mathbb{Z}^d$,  is given by
\begin{equation}
S_{\Lambda}^{\mathcal{D}}=\mathrm{Tr\;}h(P_{\Lambda}^{\mathcal{D}}),
\label{SBD}
\end{equation}%
where
\begin{equation}
h(t)=-t\log _{2}t - (1-t)\log _{2}(1-t),\; t \in [0,1]  \label{h}
\end{equation}%
is a binary Shannon entropy,
\beq \label{pld}
P_{\Lambda}^{\mathcal{D}}=\chi_\La P^{\mathcal{D}}\chi_\La
\eeq
and $P^{\mathcal{D}}$  is the Fermi projection of the restriction $H_{\mathcal{D}}=\chi_{\mathcal{D}} H\chi_{\mathcal{D}}$ of the Schr\"{o}dinger operator $H$  given in \eqref{DS} to $\mathcal D$\footnote{It is important to recognize that $P^{\mathcal{D}}$ neither coincides with $\chi_{\mathcal{D}} P\chi_{\mathcal{D}}$ (where $P$ is given in \eqref{P}) nor is $P_{\Lambda}^{\mathcal{D}}$ is equal to $P_{\Lambda}$, since $H$ is not a  multiplication operator.}. Here and below we treat the indicator
$\chi_\La$ of $\Lambda$ as the orthogonal projection from $\ell^2(\Z^d)$ to $\ell^2(\La)$, thus $P_{\Lambda}^{\mathcal{D}}=\chi_\La P^{\mathcal{D}}\chi_\La$ is the restriction
of  $P^{\mathcal{D}}$ to $\Lambda$.  We recall that the Fermi projection $Q_\mu$ of the self-adjoint operator $K$ is its spectral projection-valued  measure $\mathcal{E}_{K}$, corresponding to the interval $(-\infty
,\mu ]$, i.e., $Q_\mu=\mathcal{E}_{K}((-\infty ,\mu ])=\chi_{(-\infty,\mu]}(K)$
and  $\mu$ is the Fermi energy.

We also note  that for a Hermitian operator $K$ on $\ell^2(\Z^d)$,  $f(K)$ is defined  by means of its spectral decomposition.

It is easy to show that $H_{\mathcal{D}}$ converges strongly to $H$ as $%
\mathcal{D} \nearrow \mathbb{Z}^d$, say in the van Hove sense \cite{Ru:77},
hence $P^{\mathcal{D}}$ converges strongly to the Fermi projection
\begin{equation}\label{P}
P=\mathcal{E}_{H}((-\infty
,\mu ])
\end{equation}
of $H$ provided that  $\mu$ is not its eigenvalue. Since $\Lambda$ is a finite set, $S_{\Lambda}^{\mathcal{D}}$ of \eqref%
{SBD} converges to
\begin{equation}  \label{SB}
S_{\Lambda}=\mathrm{Tr\,} h(P_{\Lambda}),
\end{equation}
where (cf. (\ref{pld}))
\begin{equation}  \label{PL}
P_{\Lambda}=\chi_\La P\chi_\La
\end{equation}
is the restriction of $P$ of \eqref{P} to $\La$.

We will assume in this paper that the potential \eqref{pot} is an ergodic
field in $\mathbb{Z}^{d}$. Recall that the field is defined by a measurable
function $v: \Omega \to \mathbb{R}$ on a probability space \ $(\Omega ,\mathcal{F},P)$ endowed with a measure preserving and ergodic group of transformations $\{T_{a}\}_{a\in
\mathbb{Z}^{d}}$ \cite{Ko-Co:81}:%
\beq
V(\omega)=\{V(x,\omega)\}_{x \in \mathbb{Z}^d}, \; V(x,\omega )=v(T_{x}\omega ), \; x \in \mathbb{Z}^d, \; \omega\in\Omega \label{erpot}
\eeq
or
\beq
V(x,T_{a}\omega )=V(x+a,w),\;v(\omega )=V(0,\omega ), \; a, x \in \mathbb{Z}^d, \; \omega \in \Omega.  \label{verg}
\eeq%
As a result, the whole operator defined by (\ref{DS})
and (\ref{erpot}) -- (\ref{verg})
\beq \label{DSe}
H(\omega)=-\Delta +V(\omega)=\{H(x,y,\omega )\}_{{x,y}\in \Z^{d}}
\eeq
is an ergodic operator (see \cite{Pa-Fi:92}), i.e., satisfies
the relation%
\begin{equation}
H(x,y,T_{a}\omega )=H(x+a,y+a,\omega ),\; a,x,y\in \mathbb{Z}^{d},
\;  \; \omega \in \Omega.
\label{eop}
\end{equation}%
A particular case of the ergodic  Schr\"{o}dinger operator whose potential is a collection of i.i.d. random variables is known as the Anderson model.
More generally, a random operator $A(\omega)=\{A(x,y,\omega)\}_{x,y \in \mathbb{Z}^d}$ in $\ell^2(%
\mathbb{Z}^d)$ is called ergodic if it satisfies \eqref{eop}.

It then  follows (see \cite{Pa-Fi:92}, Theorem 2.7) that the Fermi projection
\eqref{P} is an ergodic orthogonal projection $P(\omega)$, i.e., it is selfadjoint, $P^{2}(\omega)=P(\omega)$ and
\begin{align}\label{epr}
&P(\omega)=\{P(x,y,\omega )\}_{x,y\in \mathbb{Z}^{d}}, \; \omega \in \Omega.
\\& \hspace{-1cm}
P(x,y,T_{a}\omega)=P(x+a,y+a,\omega ), \; a,x,y\in \mathbb{Z}^{d}, \; \omega \in \Omega.
\notag
\end{align}%
In particular, for any collection $\{(x_{i},y_{i})\}_{j=1}^k$ of  points in $\mathbb{Z}^{d}\times \mathbb{Z}^{d}$ we have:
\begin{equation}
\mathbf{E}\Big\{\prod_{i=1}^{k}P(x_{i},y_{i})\Big\}=\mathbf{E}\Big\{%
\prod_{i=1}^{k}P(x_{i}+a,y_{i}+a)\Big\},\quad \forall a\in \mathbb{Z}^{d}.
\label{trans}
\end{equation}%
Here and below the symbol $\E\set{\ldots}$ denotes the expectation in the above probability space and we omit
the event variable $\omega \in \Omega$ in expectations.

We will also denote by $\{e_n\}_{n=1}^d$ the standard basis of $\Z^d$. It will be convenient to use the maximum norm $\|x\|_\infty=\max \left(\left|x_{1}\right|,\ldots ,\left|x_{d}\right|\right)$ for a vector $x=\pa{x_1,\ldots,x_d}\in\Z^d$. It induces the distance
\beq\label{eq:dist}
\dist(\mathcal C_1,\mathcal C_2)=\min_{x\in
\mathcal C_1,y\in \mathcal C_2}\|x-y\|_\infty
\eeq
and the boundary of $\mathcal C$, which is defined by
\beq\label{partialC}
\partial \mathcal C=\set{x\in \Z^d:\ \dist(x,\mathcal C^c)=1}.
\eeq
 We will confine ourselves to the case where
the block $\Lambda$ in \eqref{SBD} -- \eqref{SB} is a $d$-dimensional lattice cube
\begin{equation}
\Lambda_M =[-M,M]^{d}\cap \mathbb{Z}^{d},\;|\Lambda |=L^{d},\;L=2M+1,
\label{Lad}
\end{equation}
where $M$ is a positive integer. Note that $\La_M=\{x\in\Z^d:\ \|x\|_\infty\le M\}.$

\medskip
A basic ingredient in proofs of our results below is
the exponential decay of the  off diagonal matrix elements of the Fermi projection given in \eqref{P} and (\ref{epr}),
if the Fermi energy $\mu$ lies in the exponentially localized part of the spectrum of  $H(\omega)$. This is a central result in the spectral theory of Schr\"{o}dinger operators with ergodic potential. Concretely, it states that for such $\mu$  there exist $C_0< \infty$ and $\gamma>0$ such that
\begin{equation}
\mathbf{E}\{|P(x,y)|\}\leq C_0e^{-\gamma |x-y|}, \label{b_P}
\end{equation}%
for all $x,y \in \mathbb Z^d$. The bound is a manifestation of the  exponential localization for the
Schr\"{o}dinger operator in the neighborhood of $\mu$, i.e.,  the pure point  spectrum
with exponentially decaying  eigenfunctions.
\medskip

Let us list some of the well known cases in which the validity of \eqref{b_P} is established for the discrete ergodic
Schr\"{o}dinger operator $H(\omega)$ given by \eqref{DSe}:
\begin{enumerate}
\item[(a)] For i.i.d. random potentials, any $d\ge 1$ and any $\mu$ in the spectrum of $H(\omega)$
if the amplitude of the potential (see e.g. (\ref{VQ})) is large enough  and its
probability distribution $F$ is regular enough (see \eqref{hoc} for the precise
definition of the required regularity).
\item[(b)] For i.i.d. random potentials, any $d\ge 1$ and  $\mu$ sufficiently close
to the spectrum edges of $H(\omega)$ if the amplitude of the potential is fixed and
$F$ is regular enough. In particular, in this case the intervals of exponential
localization  could coexist with
intervals of extended states below or above the exponentially localized part of
spectrum, which contains $\mu$.
\item [(c)] For  i.i.d. random potentials, $d=1$, any $\mu$ in the spectrum
of $H(\omega)$, any $F$  that  is  not concentrated at a
 single  point and  any amplitude of the random potential.
\item [(d)] For an interesting and widely studied case  of a non-random ergodic
(quasiperiodic) potential, where  $ V(x,\omega)=2g \cos 2\pi(\alpha x + \omega), \; x \in \mathbb{Z}, \; \omega \in [0,1)$, with  Diophantine $\alpha$ and $g>1$ (the
  supercritical almost Mathieu operator with Diophantine frequency).
\end{enumerate}
In brief, the cases (a) -- (b) and (d) describe the exponential localization
either at high disorder or at extreme energies for a $d \ge 1$ dimensional
ergodic Schr\"{o}dinger operator. Let us mention that the bound \eqref{b_P} extends to a broader class of random operators with a more general than the discrete Laplacian  (\ref{dl}) "hopping" part and/or with correlated
random potentials. The bound is closely related to the analogous  bound \eqref{dynloc} for the fractional moments of the Green function of corresponding operator.
 It is also worth noting that (at least for i.i.d. potentials) the validity of \eqref{b_P} requires the exponential localization (e.g. \eqref{dynloc}) only in a neighborhood of the Fermi energy $\mu$  but not in the whole halfline $(-\infty,\mu]$. The corresponding
physical intuition is that at zero temperature only states close to the Fermi
energy determine the properties of the free Fermi gas and the corresponding
mathematical proof (at least for i.i.d. potentials) is based on the exponential decay \eqref{dynloc} of the Green function's fractional
moments of the operator in question. We refer the reader to the works
\cite{Ai-Gr:98,Ai-Co:01,Ai-Wa:15,Ge-Ta:12,Ji-Kr:13,Mi:96,St:11} for results and references on various aspects of the validity and applications of the
bound.

Although the operator $P(\omega)$ introduced in \eqref{P} is the Fermi projection of
the ergodic Schr\"{o}dinger operator $H(\omega) $ given by (\ref{DSe}),  a considerable
amount of our results can be formulated and proved independently of the origin of
$P(\omega)$. In particular, the results concerning the mean entanglement
entropy (see Results \ref{t:2} -- \ref{t:1p1}) are valid for {\it{any}} ergodic
orthogonal projection satisfying \eqref{epr} and \eqref{b_P}
(cf.\  Assumption \ref{assump}).   For instance, it can be the spectral projection
of $H$
\beq\label{PEI}
P (\omega)= \mathcal E_{H (\omega)}(I)
\eeq
where $I$ is a subset of the spectrum of $H(\omega)$ for which the exponential bound \eqref{b_P} holds.

Given an orthogonal projection $P = \set{P(x,y)}_{x,y\in\ell^2(\Z^d)}$ and the sets $\mathcal C_1\subset \Z^d$ and $\mathcal C_2\subset \Z^d$, consider the self-adjoint operator acting on $\ell^2(\mathcal C_1)$
\beq\label{picc}
\Pi_{\mathcal C_1,\mathcal C_2}:=\chi_{\mathcal C_1}P\chi_{\mathcal C_2}P\chi_{\mathcal C_1}
=\set{\Pi_{\mathcal C_1,\mathcal C_2}(x,y)}_{x,y\in\mathcal C_1};
\eeq
\beq\label{picc1}
\Pi_{\mathcal C_1,\mathcal C_2}(x,y)=\sum_{z\in\mathcal C_2}P(x,z)P(z,y),
\eeq
see Lemma \ref{l:picc} for its properties. We will often deal with a case  $\mathcal C_2={\mathcal C_1}^c$ (with $\mathcal C^c=\Z^d\setminus \mathcal C$), where we will use the shorthand notation
\beq\label{pic}
\Pi_{\mathcal C}=\Pi_{\mathcal C,{\mathcal C}^c}.
\eeq
We will also use the following change of variables for the function $h$ of \eqref{h}:%
\begin{equation}
h(t)=h_{0}(t(1-t)),\;t\in \lbrack 0,1].  \label{h0}
\end{equation}%
This relation defines $h_0: [0,1/4] \to \ [0,1]$ implicitly;  its explicit definition of and necessary
properties are given in Lemma \ref{l:h0}.

It follows then from
\eqref{pic}
that if
$P_{\Lambda}$ is defined in (\ref{PL}), then%
\beq\label{Pi1}
P_\La(\chi_\La-P_\La)=\Pi_\La
\eeq
and
\eqref{SB} and \eqref{h0} imply
\begin{equation}\label{sh0}
S_{\Lambda }=\tr h_{0}(\Pi _{\Lambda}).
\end{equation}%
We remark that the right hand side is well defined in view of Lemma \ref{l:h0} (ii).


\subsection{Results}
We start with a simple general observation asserting that  the
large block behavior of the entanglement entropy of the ergodic system is intrinsically
different from that of the thermodynamic entropy, which is extensive, i.e., asymptotically
proportional to the volume $L^{d}$ of the block $\Lambda_{M} $ defined in \eqref{Lad}
\cite{Pa-Fi:78,Ru:77}. In addition, the proof of the assertion shows the advantage of using
the  formula \eqref{sh0} rather than \eqref{SB}, since the former explicitly takes into
account the fact that the main
contribution to $S_{\Lambda_{M} }$ comes from a sufficiently thick layer adjacent
to   the surface of $\Lambda_{M} $ - a fact that is systematically used below.
This can be seen from the  decay of the matrix elements of the operator $\Pi _{\Lambda}$
(see \eqref{pic} with $\mathcal{C} = \Lambda$) away from the boundary of $\Lambda$.

Here we will use only the  slow decay of $P(x,y)$
required by the equality
\begin{equation}\label{sr}
\sum_{y\in \mathbb{Z}^{d}}|P(x,y)|^{2}=P(x,x)\leq 1,
\end{equation}
 which is valid for any orthogonal projection in $\ell^2\pa{\Z^d}$. In subsequent
assertions, we will use the exponential bound \eqref{b_P}, which will allow us
to establish a variety of asymptotic properties of $S_\Lambda$ as the size of
$\Lambda$ tends to infinity.

\begin{result}\label{t:expvan}
Let  $ S_{\Lambda }(\omega)$ be the entanglement entropy  \eqref{SB} of
disordered free lattice fermions whose one body Hamiltonian is a discrete
ergodic Schr\"{o}dinger operator $H(\omega)$ defined in  \eqref{DSe} and
\eqref{erpot} -- \eqref{verg} and let $P(\omega)$ be its
Fermi projection defined in \eqref{P} and \eqref{epr} -- \eqref{trans}.
Then, for $\Lambda=\Lambda_M:=[-M,M]^d \subset  \mathbb{Z}^d,\; L=2M+1$, we have:
\begin{equation}
\lim_{L\rightarrow \infty }L^{-d}\mathbf{E}\{S_{\Lambda_M }\}=0.  \label{enex}
\end{equation}
\end{result}
\begin{proof}
The assertion is a special case  of Theorem \ref{thm:apriori}  (with the choice $f = h_0$
there),  proven in the next section. This can be seen from \eqref{h}, \eqref{h0},
\eqref{sh0} and  Lemma \ref{l:h0} (v),
which implies that  $h_0$ satisfies Condition \ref{c:sob} for any $\alpha\in (0, 1)$.
\end{proof}                                                                                                        %

\medskip
To present in a compact form our results on the area law in the mean,
we will assume certain symmetry properties of  the ergodic potential given in (\ref{DS}), (\ref{pot}) and (\ref{erpot}) -- (\ref{verg})
(the general case is described in Remark \ref{remd}).

Assume that in addition to the  ergodic group $\{T_{a}\}_{a\in
\mathbb{Z}}$ introduced in \eqref{erpot} -- \eqref{epr}, the probability space
is endowed with the measure preserving transformation $R$ (reflection) such that
\beq\label{ref}
V(x,R\omega)=V(-x,\omega ),\quad x\in\Z^d,\quad \omega\in\Omega.
\eeq
For instance, this is the case for any random i.i.d. potential in any dimension  as well as
for quasiperiodic potentials  $V(x,\omega )=v(\alpha x+\omega )$, $x\in\Z$, $\omega\in [0,1)$, where $v:%
[0,1) \rightarrow \mathbb{R}$ is an even 1-periodic function, $%
\alpha $ is an irrational number and $\omega $ is uniformly distributed over
the one-dimensional torus $[0,1)$.

Assume also that there exists a collection of measure preserving
transformations $\{\Sigma _{\sigma}\}_{\sigma\in S_d} $ (permutations) of the probability
space that forms a representation of the symmetric group $S_{d}$  on $d$ symbols and such that%
\begin{equation}
V(x,\Sigma_\sigma \omega )=V(\sigma x,\omega ),\;x\in \mathbb{Z}^{d},\;\sigma \in
S_{d},\;\omega\in\Omega.  \label{per}
\end{equation}%
This property is valid in the case of i.i.d. potential 
in any dimension.

Since the $d$-dimensional discrete Laplacian commutes with the reflection $%
x\rightarrow -x$ and permutations of the components $x=(x_{1},...,x_{d})%
\rightarrow \sigma x=(x_{\sigma (1)},...,x_{\sigma (d)})$ of vectors in $%
\mathbb{Z}^{d}$, the Schr\"{o}dinger operator   (\ref{DSe}) and consequently its Fermi
projection (\ref{P}) also possesses these
properties (see Theorem 2.7 of \cite{Pa-Fi:92}):%
\beq\label{refp}
P(x,y,R\omega )=P(-x,-y,\omega ),\;x,y\in \mathbb{Z}^{d} 
, \;  \omega \in \Omega \eeq%
and%
\beq\label{perp}
P(x,y,\Sigma _{\sigma }\omega )=P(\sigma x,\sigma y,\omega ),\;x,y\in
\mathbb{Z}^{d}, \;  \omega \in \Omega .
\eeq

To formulate our second assertion, we introduce the following notation:
\beq\label{zpmd}
\Z_\pm=[0,\pm\infty)\cap \Z,\quad \Z_\pm^d=\Z_\pm\times\Z^{d-1}.
\eeq
We also remind the reader that $\{e_n\}_{n=1}^d$ stands for the standard basis of $\Z^d$.
\begin{result}\label{t:2}
Let  $ S_{\Lambda }(\omega)$ be the entanglement entropy  \eqref{SB} of
disordered free lattice fermions whose one body Hamiltonian is a discrete
ergodic Schr\"{o}dinger operator $H(\omega)$ defined in  \eqref{DSe} and
\eqref{erpot} -- \eqref{verg} and let $P(\omega)$ be its Fermi projection
defined in \eqref{P}, \eqref{epr}  -- \eqref{trans} and \eqref{refp} -- \eqref{perp}.
Assume that the Fermi energy $\mu$ lies in the exponentially localized
part of the spectrum of $H(\omega)$, i.e.,  the bound \eqref{b_P} holds.
Then, for $\Lambda=\Lambda_M:=[-M,M]^d \subset \mathbb{Z}^d, \; L=2M+1$
and $B(je_1,je_1)=\langle\delta_{je_1},B\delta_{je_1}\rangle$, we have:
 \begin{equation}
\lim_{L\rightarrow \infty }L^{-(d-1)}\mathbf{E}\{S_{\Lambda_M
}\}=2d\sum_{j\in\Z_+}\mathbf{E}\{(h(P_{\mathbb{Z}%
_{+}^{d}}))(je_1,je_1)\}<\infty,  \label{t2}
\end{equation}%
where $P_{\mathbb{ Z}^d_+}(\omega)=\chi_{\mathbb{ Z}^d_+}P(\omega)
\chi_{\mathbb{ Z}^d_+}$.

\end{result}
\begin{proof}
The assertion is a special case  of Theorem \ref{thm:splitinf}
(with the choice $f = h_0$ there),  proven in the next section.
This can be seen from \eqref{h}, \eqref{h0},  \eqref{sh0} and
Lemma \ref{l:h0} (v), which implies that $h_0$ satisfies
Condition \ref{c:sob} for any $\alpha\in (0, 1)$.

Note that the Fermi projection satisfies the bound (\ref{b_P}) under the conditions given
in items (a) --  (d) of the list below (\ref{b_P}) and the text below the list.
\end{proof}

\begin{remark}
\label{remd} We will present here the general form of Result \ref{t:2} where
we do not assume the symmetry properties \eqref{refp} -- \eqref{perp} of the underlying
ergodic projection. Denote%
\begin{equation*}
\mathbb{Z}_{s}^{d}(j)=\{x=(x_{1},...,x_{d})\in \mathbb{Z}^{d}:sx_{j}\geq
0\}, \; s=\pm.
\end{equation*}%
Then instead of  \eqref{t2} we have
\beq\label{t2g}
 \lim_{L\rightarrow \infty }L^{-(d-1)}\mathbf{E}\{S_{\Lambda_M
}\}=\sum_{j=1}^{d}\sum_{s=\pm }\sum_{x_{j}\in \mathbb{Z}^d_s(j)}
 \mathbf{E}\{(h(P_{\mathbb{Z}%
_{s}^{d}(j)}))(x_{j}e_j,x_je_j)\}.
\eeq%
\end{remark}
In particular,    in the one dimensional case, one has
\beq\label{lim1}
 \lim_{L\rightarrow \infty }\mathbf{E}\{ S_{\Lambda_{M}
}\}=\E\set{\tr h\pa{P_{\Z_-}}}+\E\set{\tr h\pa{P_{\Z_+}}},
\eeq
with $\mathbb{Z}_\pm$ defined in  \eqref{zpmd}.

\begin{remark}\label{r:trinv}
It is interesting to compare the above results with those in the translation
invariant case, where the operator $A$ in \eqref{HF} is a
convolution operator in $\ell^2(\mathbb{Z}^d)$:
\begin{equation}\label{conv}
A=\{A(x-y)\}_{x,y \in \mathbb{Z}^d}, \; \;\sum_{x \in \mathbb{Z}^d}
|A(x)| < \infty.
\end{equation}
In this case,  for $d=1$,
the Fermi projection (\ref{P}) is
\begin{equation}\label{pconv}
P(x,y)=\frac{\sin p(\mu) (x-y)}{\pi\pa{x-y}},
\end{equation}
where $p(\mu) \in [0, \pi]$ is the Fermi momentum, determined by  the
Fourier transform (symbol) of $\{A(x)\}_{ x\in \mathbb{Z}^d}$ (cf. \eqref{P}
and \eqref{PEI}). It follows
then from   \eqref{SB}, Lemma \ref{l:h0} (iv) and (\ref{pconv}) that for $d=1$ we have
\begin{align}\label{log1}
 S_{\Lambda_{M}} \ge 4 \sum_{|x|\le M, |y| > M} |P(x,y)|^{2} =
 \frac{4}{\pi^2}\log L +O(1), \; L \to \infty.
\end{align}
A similar argument  for $d > 1$ yields
\begin{align}\label{logd}
 S_{\Lambda_{M}}\ge C_d L^{d-1} \log L +O(L^{d-1}), \; L \to \infty.
\end{align}

The bounds \eqref{log1} -- \eqref{logd} provide a simple
manifestation of  logarithmic corrections to the area law in
translation invariant macroscopic systems
\cite{Ca-Co:11,Ei-Co:11,Gi-Kl:06,He-Co:09,Le-Co:13,So:13,Wo:06}.
Moreover, these bounds emphasize the difference between the translation
invariant and disordered cases.
\end{remark}

In the one dimensional case, it is possible not only to prove the
existence of the limit of the mean  entanglement entropy, i.e., to find
the leading term of the asymptotics of the mean entanglement entropy
as $L\rightarrow\infty$, but also to find the leading term for all
typical realizations, i.e., with probability $1$. This can be viewed
as the one dimensional version of the area law for typical realizations.
\begin{result}
\label{t:1p1}
Let  $ S_{\Lambda }(\omega)$ be the entanglement entropy  \eqref{SB} of
disordered free lattice fermions in dimension 1 whose one body Hamiltonian is a
discrete ergodic Schr\"{o}dinger operator $H(\omega)$ defined in  \eqref{DSe}
and \eqref{erpot} -- \eqref{verg} and let
$P(\omega)$ be its Fermi projection defined in \eqref{P} and
\eqref{epr} -- \eqref{trans}.  Assume that
the Fermi energy $\mu$ lies in the exponentially localized part of the
spectrum of $H(\omega)$, i.e., the bound \eqref{b_P} holds. Then, for
$\La=\La_M:=[-M,M] \subset \mathbb{Z}, \; L=2M+1$,  we have,
with probability 1:

\begin{equation}
S_{\Lambda }(\omega)=S_{+}(T_{M}\omega )+S_{-}(T_{-M}\omega
)+o(1),\;L:=(2M+1)\rightarrow \infty,   \label{sp1}
\end{equation}%
where
\begin{equation}
S_{\pm }(\omega)=\mathrm{Tr\ }h(P_{\mathbb{Z}_{\mp }}(\omega)),
\quad P_{\Z_\pm}(\omega)=\chi_{\Z_\pm}P(\omega)\chi_{\Z_\pm}, \label{Spm}
\end{equation}%
$\Z_\pm$  is defined in \eqref{zpmd}  and the shift ergodic transformations
$T_{\pm M}$ are defined in \eqref{erpot} -- \eqref{epr}.

The random variables  (\ref{Spm}) are finite and not  zero if the Fermi energy
$\mu$ in (\ref{P}) -- (\ref{SB}) lies strictly inside the spectrum (equivalently,
if $P$ is neither the zero operator nor the identity operator).

\end{result}
\begin{proof}
The assertions on the existence of a well-defined (i.e., finite with
probability $1$) 
asymptotics \eqref{sp1} -- \eqref{Spm} 
follow from those of Theorem \ref{t:fp1}  (with the choice $f=h_0$ there).
This can be seen from \eqref{h}, \eqref{h0},  \eqref{sh0}  and  Lemma \ref{l:h0} (v)
which implies that $h_0$ satisfies Condition \ref{c:sob} for any $\alpha\in (0, 1)$.

Let us prove that the random variables \eqref{Spm}
are not zero with probability $1$. Indeed, assume that $S_+ (\omega)= 0$ with
probability $1$ (the case $S_- (\omega)= 0$   can be considered analogously).
It follows from \eqref{Spm},
\eqref{h0}, \eqref{sh0} and Lemma \ref{l:h0} (ii) that $\Pi_{\Z_+} (\omega)= 0$,
with probability 1. Now, taking into account \eqref{picc} -- \eqref{picc1} and the
fact that $P(\omega)$ is hermitian and ergodic, we obtain that $P(x,y,\omega) = 0$
for $x\neq y$ with probability $1$, i.e., that the projection  $P(\omega)$ is diagonal:
$P (x, y,\omega) = P (x, x,\omega)\delta_{xy}$. Since $P(\omega)$ commutes with the
Schr\"{o}dinger operator \eqref{DSe}, we have $P(x,x,\omega) = P(x+1,x+1,\omega)$,
$\forall x  \in\Z$, i.e., $P(\omega) = p(\omega){\bf 1}_{\Z}$, where
$p(\omega) \in \set{0,1}$ and $p(\omega)=p(T_a\omega), \; \forall
a \in \mathbb{Z}$. Since the group $\{T_a\}_{a \in \mathbb{Z}}$ is ergodic,
i.e., has no invariant subsets in $\Omega$ except $\emptyset$ and $\Omega$,
$p$ is independent of $\omega$. Thus, $P(\omega)$ is either the zero operator
or the identity operator, contrary to our assumption that $P$ is a non-trivial
projection.

Note that the Fermi projection satisfies the bound (\ref{b_P}) under the conditions given
in items (a) --  (d) of the list below (\ref{b_P}) and the text below the list.
\end{proof}
%

\medskip

Note that the most studied class of  operators for which \eqref{b_P} holds
consists of Schr\"{o}dinger operators with i.i.d. potential or more generally,
potentials with sufficiently fast decay of statistical correlations,
see  \cite{Ai-Wa:15,Ge-Ta:12,St:11} and the items (a) -- (c) of the list after
formula (\ref{b_P}). However, the bound \eqref{b_P} also holds for one dimensional
Schr\"{o}dinger operators with quasiperiodic potentials (see, e.g., \cite{Ji-Kr:13}
and the item (d) of the list after  (\ref{b_P})), which have, so to speak, a minimal
amount of randomness. This shows that the hypotheses of Results \ref{t:2} and \ref{t:1p1}
above can be satisfied even for systems where statistical
correlations of the associated ergodic potential do not exhibit fast decay.
Our next assertion on  the power law decay (in $L$) of the variance  of the entanglement
entropy per unit area for $d \ge 2$, however,  does rely on independence of random
potentials in the corresponding discrete Schr\"{o}dinger operator (Anderson model).
\begin{result}\label{t:var}
Let  $ S_{\Lambda }(\omega)$ be the entanglement entropy  \eqref{SB} of
disordered free lattice fermions whose one body Hamiltonian is a discrete ergodic Schr\"odinger
operator $H(\omega)$ defined in  \eqref{DSe} and \eqref{erpot} -- \eqref{verg} and
let $P(\omega)$ be its Fermi projection defined \eqref{P} and \eqref{epr} -- \eqref{trans}.
Assume  that the potential  in $H(\omega)$ is a collection of i.i.d. random
variables such that their common probability distribution $F$ is uniformly
H\"{o}lder continuous:
\beq \label{hoc}
F ((v - \eps, v + \eps)) \le C |\eps|^\tau, \; \forall v \in \supp F, \; \eps > 0,\;
\tau \in (0,1],
\eeq
where $C$ is independent of $v$. Assume also that
the Fermi energy $\mu$ lies in the exponentially localized part of the
spectrum of $H(\omega)$, i.e., the bound \eqref{b_P} holds. Then,
for $\La=\La_M:=[-M,M]^d \subset \mathbb{Z}^d, \; L=2M+1$,
we have:
\begin{align}\label{eq:varmain}
&\Var\left\{L^{-(d-1)}\,S_{\La_M}\right\}:=
\E\set{\pa{L^{-(d-1)}\,S_{\La_M}}^2}-\pa{\E\set{ L^{-(d-1)}\,S_{\La_M}}}^2
\\& \hspace{4.5cm} \le C\pa{\log L}^{4d/(d+1)}L^{-2(d-1)/(d+1)},\; \; L \to \infty.
\notag
\end{align}
\end{result}
\begin{proof}
The assertion is a corollary  of Theorem \ref{thm:varia}  with the choice $f = h_0$ there
and Lemma \ref{l:iidpr}. This can be seen from \eqref{h}, \eqref{h0},  \eqref{sh0}
and  Lemma \ref{l:h0} (v) which implies that $h_0$ satisfies Condition \ref{c:sob}
for any $\alpha\in (0, 1)$.\end{proof}

Note that the Fermi projection satisfies the bound (\ref{b_P}) under the conditions given
in items (a) --  (c) of the list below (\ref{b_P}) and the text below the list.

\section{Proofs}
In this section we prove several assertions  that  are more general versions
of Results \ref{t:expvan} - \ref{t:var} of  the previous section. They are valid for a certain class of functions that includes the function $h$ described in \eqref{h} and for a
class of ergodic operators that includes  a Schr\"{o}dinger operator with an ergodic potential. Namely, the assertions of this section focus on the
large block behavior of the quantity
\beq\label{eq:FM}
F_\La(\omega)=\tr f\pa{\Pi_\La (\omega)}
\eeq
where $\Pi_\La ((\omega))$ is defined in \eqref{pic}, \eqref{Pi1} and Assumption
\ref{assump} below and $f$ satisfies
\begin{condition}\label{c:sob}
The function $f$ is supported on $[0,1/4]$ and $f\in C^2(0,1/4]$. Moreover, there exists $\alpha\in(0,1]$  such that
\[
\max_{k\in\{0,1,2\}}\sup_{x\neq0}\abs{f^{(k)}(x)}\abs{x}^{k-\alpha}<\infty.
\]
\end{condition}
Note that the right hand side in \eqref{eq:FM} is well defined in view
Lemma \ref{l:picc} (ii).

For a function $f$ satisfying Condition \ref{c:sob}, the bound
\beq\label{f}
\abs{f(x)}\le C\abs{x}^{\alpha}
\eeq
holds uniformly for all $x\in[0,1/4]$.

 We remark that the function $h_0$ of \eqref{h0}  belongs to this class for any $\alpha\in (0, 1)$, see Lemma \ref{l:h0}
(v).  Thus, the entanglement entropy \eqref{sh0} is a particular case of \eqref{eq:FM} -- \eqref{f}.

We will also consider  a more general class of ergodic orthogonal projections.
Namely, we will not assume below that the
orthogonal projection in \eqref{pic} and \eqref{sh0} is  the Fermi  projection \eqref{P} (or  the spectral projection \eqref{PEI})
of a Schr\"{o}dinger operator \eqref{DSe} with an ergodic potential.  Instead we will require the following properties of projections.
\begin{assumption}\label{assump}
The orthogonal projection $P(\omega)$
 is ergodic (i.e., satisfies \eqref{epr} -- \eqref{trans}) and
\beq\label{eq:assump}
\E\{\abs{P(x,y)}\}\le C\e^{-\gamma\abs{x-y}},\quad \forall \ x,y\in\Z^d,
\;\gamma>0, \; C < \infty.
\eeq
\end{assumption}
This assumption is satisfied for the Fermi projection (\ref{P}) of the
Schr\"{o}dinger
operator (\ref{DSe}) with ergodic potential (\ref{erpot}) -- (\ref{verg}) in situations described in the items (a) -- (d) of the list below \eqref{b_P}.

\medskip
The i.i.d. randomness requirement needed for Result \ref{t:var} is now replaced by
\begin{assumption}\label{assump1}
Suppose that $P(\omega)$ is an ergodic orthogonal projection satisfying Assumption \ref{assump} and, in addition, that for  any finite $\La\subset\Z^d$
there exists a random orthogonal projection $\widehat P_\La(\omega)=\{\widehat P_\La(x,y,\omega)\}_{x,y
\in \La}$ in $\ell^2(\La)$ with the following properties:
\begin{enumerate}
\item[(i)] Proximity to $P(\omega)$ away from the boundary of $\Lambda$:
\beq\label{eq:rSch}
\E\Big\{\abs{P(x,y)-\widehat P_\La(x,y)}\Big\}\le C\abs{\partial\La}\e^{-\widetilde \gamma R}, \; \widetilde{\gamma}>0, \; C <\infty
\eeq
for any $(x,y)\in\La\times\La$ that satisfies
\beq\label{ddl}
\dist\pa{\{x\},\partial\La}+\dist\pa{\{y\},\partial\La}\ge R,
\eeq see \eqref{eq:dist} -- \eqref{partialC} for the notation;

\item[(ii)] Statistical independence: For $\La\subset\Z^d$ and $\mathcal K,\mathcal L \subset \La$ consider a rigid motion (i.e., a composition
of  translation and rotation in $\Z^d$ ) $g$ that takes $\La$ to $\La _g$, $\mathcal K$ to $\mathcal K_g$, and $
\mathcal L$ to  $\mathcal L_g$. Then the  random variables
\beq\label{xih}
\xi(\omega)=\tr f\pa{ \chi_{\mathcal K} \widehat P_\La (\omega)\chi_{\mathcal L}\widehat P_\La (\omega)\chi_{\mathcal K}},\quad \xi_g(\omega)=\tr f\pa{ \chi_{\mathcal K_g} \widehat P_{\La_g}(\omega) \chi_{\mathcal L_g}\widehat P_{\La_g} (\omega)\chi_{\mathcal K_g}}
\eeq
are independent and identically distributed, as long as $\La\cap\La_g=\emptyset$.
\end{enumerate}
\end{assumption}
 In Lemma \ref{l:iidpr} we show that the conditions of Assumption \ref{assump1} are met for the discrete Schr\"{o}dinger operator (\ref{DS}) with an i.i.d.
potential (Anderson model). In this case $\widehat P_\La$ is  the Fermi projection for the restriction $H_\Lambda(\omega)$ of $H(\omega)$ to $\Lambda$.

We remark that for all but one  results in this section, we will use only Assumption \ref{assump}. 
\begin{theorem}\label{thm:apriori}
Let $F_{\La}(\omega)$ be defined by \eqref{eq:FM} and (\ref{picc}) -- (\ref{pic}),
where  $f$ satisfies Condition \ref{c:sob} and $P(\omega)$ is an ergodic projection (see  \eqref{epr} -- \eqref{trans}). Then, for $\Lambda=\Lambda_M:=[-M,M]^d \subset \mathbb{Z}^d,\; L=2M+1$, we have:
\[
\lim_{L\rightarrow\infty}L^{-d}\,\E\,\set{\abs{F_{\La_M}}}=0.
\]
\end{theorem}
\begin{proof}
Fix $x\in\Z^d$ and set
\beq\label{gR}
u(R)=\E \,\Big\{ \sum_{y\in\Z^d:\; |x-y|\ge R}\abs{P(x,y)}^2\Big\}.
\eeq
Clearly, $u$ is  monotone decreasing and is independent of $x$ by ergodicity of $P(\omega)$, see (\ref{trans}). Moreover, we have from \eqref{sr}:
\beq\label{eq:gbeh}
u(R)\le 1,\quad \lim_{R\rightarrow \infty}u(R)=0.
\eeq
By  \eqref{eq:decf} (with $m = 1$)
\[\abs{f\pa{\Pi_{\La_M}}(x,x,\omega)}\le C\Big\{\sum_{y\in \La_M^c}\abs{P(x,y,\omega)}^{2}\Big\}^\alpha,\]
so using the H\"{o}lder inequality for expectations  we get
\beq\label{efmp}
\E\set{\abs{f\pa{\Pi_{\La_M}}(x,x)}} \le C \Big(\sum_{y\in \La_M^c}\E \{\abs{P(x,y)}^{2}\}\Big)^\alpha.
\eeq
This and \eqref{eq:FM} imply
\begin{align*}
\E\set{\abs{F_{\La_M}}}\le \sum_{x\in \La_M}\E\set{\abs{f\pa{\Pi_{\La_M}}(x,x)}}
 \le C  \sum_{x\in \La_M}\Big(\sum_{y \in \La_{M^c}} \E\{\abs{P(x,y)}^{2}\}
\Big)^\alpha.
\end{align*}
Now, we choose a positive integer $\ell<M$ whose value we will set later and split the sum over $x\in\La_M$ on the right hand side into $\Sigma_1+\Sigma_2$, where
$\Sigma_1$ is the sum over  $x\in \La_{M-\ell}$  and
$\Sigma_2$ is the sum over  $x\in\La_M\setminus \La_{M-\ell}$.

If $x\in \La_{M-\ell}$, we have
\[\sum_{y\in \La_M^c}\E\set{\abs{P(x,y)}^{2}}\le
\sum_{y\in\Z^d: \; \abs{x-y}\le\ell}\E\set{\abs{P(x,y)}^{2}},\]
so it follows from \eqref{gR} that $
\Sigma_1\le CL^du^\alpha(\ell)$, for $\quad L=2M+1$.
On the other hand, in view of \eqref{sr}, the right hand side of \eqref{efmp} is bounded by $C$, so
$
\Sigma_2\le C\ell L^{d-1}.
$
Combining the two last bounds, we obtain
\[L^{-d}\,\E\,\set{\abs{F_{\La_M}}}\le C(u^\alpha(\ell)+\ell/L).\]
The choice  $\ell \to \infty, \; \ell=o(L)$  as $L\rightarrow\infty$ gives the desired result, thanks to \eqref{eq:gbeh}.
\end{proof}

\medskip
Next, we prove that the limit of the ratio of the "generalized entanglement entropy" \eqref{eq:FM} to the surface area $L^{d-1}$ of $\La_M$ exists and is finite. To avoid cumbersome formulas, we will
again confine ourselves to the case of ergodic projections satisfying \eqref{refp} -- \eqref{perp}.

We will need the following collection of subsets  of $\mathbb{Z}^d$ for $d
\ge 2$.
Let
\beq\label{fac}
\mathcal F_M^{(n)}=\begin{cases} \set{x\in\La_M:\ x\cdot e_n=M}, &1\le n\le d, \\
\set{x\in\La_M:\ x\cdot e_{|n|}=-M} & -d\le n\le -1 \end{cases}
\eeq
be the faces of the cube $\La_M$ defined in  \eqref{Lad}. Fix  $\ell\in\N$ and consider the truncation $\widehat {\mathcal F}_M^{(n)}$ of $\mathcal F_M^{(n)}$, defined by removing points that are close to the edges of $\mathcal{F}_M^{(n)}$:
\beq\label{fach}
\widehat {\mathcal F}_M^{(n)}=\left\{x\in {\mathcal F}_M^{(n)}:\min _{k\neq n}\dist\pa{x,{\mathcal F}_M^{(k)}}\ge 3\ell\right\}.
\eeq
Furthermore, we define the surface layers ${\mathcal B}_M^{(n)}$ in $\La_M$ to be the $\ell$-neighborhoods of $\widehat {\mathcal F}_M^{(n)}$, for $n=\pm
1,...,\pm d$:
\beq\label{eq:bufferset}
{\mathcal B}_M^{(n)}=\left\{x\in \La_M:\ \dist\pa{x,\widehat {\mathcal F}_M^{(n)}}< \ell\right\}, \; n=\pm 1, \pm2,..,\pm d.
\eeq
By construction, we have
\beq\label{eq:buffer}
\dist\pa{\mathcal B_M^{(n)},\mathcal B_M^{(k)}}\ge\ell \mbox{ for } n\neq k.
\eeq
Note also that the sets $\mathcal B_M^{(n)}$  are  rectangular lattice  prisms generated by rigid motion (i.e., by a collection of lattice translations and rotations) of the lattice prism
\beq\label{eq:last}
\mathcal{L}_{\ell,M}= \pa{[0,\ell)\times[-M+2\ell,M-2\ell]^{d-1}}\cap \Z^d.
\eeq
For $d=1$, we will use the following sets instead of those of above:
$\mathcal F_M^{(\pm 1)}=\widehat{\mathcal{F}}_M^{(\pm 1)}=\{\pm M\}$ and
\beq\label{cb1}
\mathcal{B}_M^{(\pm 1)}=[\pm M, \pm (M- \ell)].
\eeq
Notice that $\mathcal{B}_M^{(\pm 1)}$ are the translations of the lattice interval (cf.\ \eqref{eq:last})
\beq\label{eq:last1}
\mathcal{L}_{\ell}=[0,\ell]\cap \Z^d.
\eeq
We will also need the lattice halfspaces   $\mathbb{Z}^d_{+,n}$, for $ n=\pm 1,...,\pm d$, which
are rigid motions of the lattice halfspace  $\mathbb{Z}^d_{+}$ of  \eqref{zpmd} such that $\La_M \subset \mathbb{Z}^d_{+,n} $
and $\mathcal{F}^{(n)} \subset \partial\mathbb{Z}^d_{+,n}$. In particular, if we set
\beq
\mathbb{Z}^d_{\pm M}=(\mathbb{Z}_{\mp} \pm\{M\})\times \mathbb{Z}^{d-1},
\eeq
then we have
\beq\label{LM}
\mathbb{Z}^d_{+,\pm1}=\mathbb{Z}^d_{\pm M}.
\eeq
In preparation for our next assertion, for any $M\in\N$, let us set
\beq\label{l12}
\ell=\begin{cases} [M/4]+1 & \mbox{if }d=1 \\
[c\ln M] & \mbox{if }d>1 \end{cases},
\eeq
where $[x]$ denotes the  integer part of $x$.
We also define, for any $\mathcal{C}\subset\Z^d$,
\beq\label{Phi}
\Phi_\mathcal{C}(\omega)=\tr \chi_\mathcal{C} f \pa{\Pi_{\mathbb{Z}^d_{+}}(\omega)}\chi_\mathcal{C}.
\eeq
\begin{theorem}\label{thm:splitinf}

Let $F_{\La}(\omega)$ be defined by \eqref{eq:FM}  and (\ref{picc}) -- (\ref{pic}),
where  $f$ satisfies Condition \ref{c:sob}, $P(\omega)$ satisfies  \eqref{refp} -- \eqref{perp} and Assumption \ref{assump}. We have,  for  $\Lambda=\Lambda_M := [-M,M]^d \subset
\mathbb{Z}^d, \; L=2M+1$
\begin{enumerate}
\item[(i)] $d=1$:
\begin{align}\label{eq:pico}
\lim_{L\rightarrow\infty}\E\,\set{F_{\La_M}}=2\lim_{L\rightarrow\infty}\E\,\set{\Phi_{\mathcal{L}_{\ell}}}\\
=2\,\E\,\set{\tr f\pa{\Pi_{\Z_{+}}}}<\infty, \notag
\end{align}
where  $\Phi_{\mathcal{L}_{\ell}}(\omega)$ is defined by (\ref{Phi}) with $\mathcal{C}=\mathcal{L}_{\ell}$, $\mathcal{L}_{\ell}$ is defined in
$\eqref{eq:last1}$  and  $\ell$ is defined in  \eqref{l12};
\item[(ii)] $d \ge 2$:
\begin{align}\label{eq:pico1}
\lim_{L\rightarrow\infty}L^{-(d-1)}\,\E\,\set{F_{\La_M}}=2d\lim_{L\rightarrow\infty}L^{-(d-1)}\,\E\,\set{\Phi_{\mathcal{L}_{\ell,M}}}\\ =2d\sum_{j\in\Z_+}\E\,\set{f\pa{\Pi_{\Z^d_{+}}}(je_1,je_1)}<\infty, \notag
\end{align}
where  $\Phi_{\mathcal{L}_{\ell,M}}(\omega)$ is defined by (\ref{Phi}) with $\mathcal{C}=\mathcal{L}_{\ell,M}$, $\mathcal{L}_{\ell,M}$  is defined in $\eqref{eq:last}$,  $e_1$ is the first vector of the canonical basis $\set{e_n}^d_{n=1}$ of $\mathbb{Z}^d$ and $\ell$ is defined in  \eqref{l12} with $c$  sufficiently large but $M$-independent.
\end{enumerate}
\end{theorem}
\begin{proof}
We note first that the properties \eqref{refp} -- \eqref{perp} are assumed
just to make
the formulation of the theorem and its proof more transparent. In fact, the results of
the theorem as well as its proof can be extended to the general case of projections,
which satisfy only Assumption \ref{assump} but not \eqref{refp} -- \eqref{perp}. In this case the assertion of
the theorem is analogous to that in Remark \ref{remd}, with  $h$ and $S
_{\La_M}$ there replaced by $f$ and $F_{\La_M}$. The  proof for this extension essentially coincides with  the one given below, though it is more tedious.

We start with the proof of the second equality in \eqref{eq:pico1} (the proof of the second equality in \eqref{eq:pico} is analogous).
Let $T_{a}$ be the measure preserving shift transformation (see \eqref{erpot} -- \eqref{epr}) by a vector $a\in \mathbb{Z}^{d}$. As usual, we will denote by $A(x,y)$ the matrix elements for the operator $A$ on $\ell^2(\Z^d)$, and will write $A(x,y;\omega)$ whenever we want to stress the dependence of $A$ on the random configuration $\omega\in\Omega$. If  $a$ is orthogonal to $e_1$ we have $T_{a}\Z_+^d =\Z_+^d$, so for any pair $(x,y)\in \Z^d_+\times\Z^d_+$ we have
\begin{equation}
\Pi _{\mathbb{Z}_{+}^{d}}(x,y;T_{a}\omega )=\Pi _{\mathbb{Z}_{+}^{d}}(x+a,y+a ;\omega )  \label{Pid1}
\end{equation}%
for such $a$, with probability $1$.
It follows then from an extended
version of Theorem 2.7 of \cite{Pa-Fi:92} that  the operator $f(\Pi _{\mathbb{Z}_{+}^{d}}(\omega))$ has the same property. In particular, since $T_{a}$ is a measure
preserving transformation of the event space, we obtain that $\mathbf{E}%
\{f(\Pi _{\mathbb{Z}_{+}^{d}})(x,y\}$ does not depend on $x-\langle x,e_1\rangle e_1$ and so,
\begin{equation}
L^{-(d-1)}\E\,\set{\Phi_{\mathcal{L}_{\ell}}}=L^{-(d-1)}\pa{2M-4\ell+1}^{d-1}\sum_{0\leq j\leq \ell}\mathbf{E}\{f(\Pi _{\mathbb{Z}_{+}^{d}})(je_1,je_1)\},  \label{p2.03}
\end{equation}%
in view of \eqref{eq:last}. To verify the existence of the finite limit as $M\rightarrow \infty $ on the right hand side, we use  Lemma \ref{lem:keytech} (i) with
\beq\label{zpd}
\mathcal{C}_1=\mathbb{Z}^d_{+}=\{x \in \mathbb{Z}^d:\ \langle x,e_1\rangle \geq 0 \},\;
\mathcal{C}_2=(\mathbb{Z}^d_{+})^c=\mathbb{N}^d_{-}:=\{x \in \mathbb{Z}^d\ \langle x,e_1\rangle \leq -1 \}
\eeq
to get the bound
\beq \label{efin}
\mathbf{E}\{|f(\Pi _{\mathbb{Z}_{+}^{d}})(je_1,je_1)|\}\leq Ce^{-\alpha
\gamma j},\quad j\in\Z_+.
\eeq
The bound and (\ref{l12}) yield the second  equality in \eqref{eq:pico1}. Besides, (\ref{efin}) implies the finiteness of the limits in \eqref{eq:pico} and \eqref{eq:pico1}.

To prove  the first equality in \eqref{eq:pico} and \eqref{eq:pico1} we note that by  Lemma \ref{lem:keytech}  (see  Assumption \ref{assump} and \eqref{picc} for the notation used), contributions to
\[
\E \{F_{\La_M}\} = \E\set{\tr f \pa{\Pi_{\La_M}}}=\sum_{x \in \La_M} \E\set{f \pa{\Pi_{\La_M}}(x,x)}
\]
due to points $x \in \La_M$ that lie away from its boundary
decay exponentially in $\dist\pa{x,\partial\La_M}$.
In dimensions higher than one, i.e., for (\ref{eq:pico}), a closer inspection shows that we may neglect contribution associated with points near the boundary of $\La$, as long as their number does not exceed $o\pa{L^{d-1}},  \; L \to \infty$. Indeed, the contributions of this order wash out once we take $\lim_{L\rightarrow\infty} L^{-(d-1)} \E \{F_{\La_M}\}$. Thus, it is not surprising that in the limit $L\rightarrow\infty$ the resulting expression converges (up to a factor $2d$ originating from the number of faces in $\La_M$) to that generated by $\Z^d_+$ rather than  $\La_M$, since locally the boundary of $\La_M$ looks indistinguishable from a hyperplane $\partial \Z^d_+$.

To implement this observation in the proof of the first equality in \eqref{eq:pico} and \eqref{eq:pico1} we use  Lemmas \ref{lem:buffer} -- \ref{lem:compinf}  to get the bound
\beq\label{eq:finalspl}
\E\,\Big\{\Big|\tr f\pa{\Pi_{\La_M}}-\sum_{n=-d }^d\tr \chi_{\mathcal B_M^{(n)}}f\pa{\Pi_{\Z^d_{+,n}}}\chi_{\mathcal B_M^{(n)}}\Big| ^2 \Big\}\le C\,R_d(M),
\eeq
where
$\mathcal {B}_M^{(n)}$ and $\Z^d_{+,n}$ are defined in \eqref{fac} --
\eqref{LM}
and
\beq\label{RM}
R_d(M)=\big(M^d e^{-\alpha \gamma \ell/2}+(d-1)\ell^2M^{d-2}\big)^{2}.
\eeq
In particular, for $d=1$, we substitute $\ell$ of \eqref{l12} into \eqref{cb1} to obtain
\beq\label{edec}
\E\,\Big\{\Big|\tr f\pa{\Pi_{\La_M}}-\sum_{n=\pm 1}\tr \chi_{\mathcal B_M^{(n)}}f\pa{\Pi_{\Z_{+,n}}}\chi_{\mathcal B_M^{(n)}}\Big|^2\Big\}\le C\,\e^{-c M},\quad c>0.
\eeq
By ergodicity and \eqref{refp},
\[
\E\,\Big\{\tr \chi_{\mathcal{B}_M^{(n)}}f\pa{\Pi_{\Z_{+,n}}}\chi_{\mathcal B_M^{(n)}}\Big\}
\]
does not depend on $n=\pm 1$ and coincides with $\E\set{\Phi_{[0,[M/4]]}}$.
This and the second equality in \eqref{eq:pico}
yield the first equality in \eqref{eq:pico}.

To obtain the existence of the first limit in \eqref{eq:pico1},
we use \eqref{eq:finalspl}, choosing  $\ell$ in \eqref{eq:bufferset} as in \eqref{l12} to  balance out the two terms in \eqref{RM}. This leads to the second equality in \eqref{eq:pico1}, by the same argument as in the proof of  \eqref{eq:pico}.
\end{proof}

\medskip
The exponential decay in \eqref{edec} will play an important role in the proof of our next assertion for $d=1$.
\begin{theorem}\label{t:fp1}
Let $F_{\La}(\omega)$ be defined by \eqref{eq:FM}  and (\ref{picc}) -- (\ref{pic}),
where  $d=1$,  $f$ satisfies Condition \ref{c:sob}, $P(\omega)$ satisfies Assumption \ref{assump} and  $\Lambda=\Lambda_M := [-M,M]\subset
\mathbb{Z}$. Set
 \beq\label{Fpm}
F_\pm(\omega)=\tr f\pa{P_{\Z_\mp}(\omega)},
\eeq
where $P_{\Z_\pm}(\omega)$ are defined by \eqref{PL} with $\La=\Z_{\pm}$ and $\Z_\pm$  given by \eqref{zpmd}.
Then  $F_\pm(\omega)$ is finite with probability 1 and   we have, with the same
probability,

\beq\label{fp1}
F_{\La_{M}}(\omega)=F_+\pa{T_{+M}\omega}+F_-\pa{T_{-M}\omega}+o(1), \; L=2M+1\rightarrow\infty,
\eeq
where $T_{\pm M}$ are the ergodic shift transformations (see  \eqref{erpot} --  \eqref{epr}).

\end{theorem}
\begin{proof}
The starting point is  the bound \eqref{edec} obtained for $d = 1$, with $\ell$ given in \eqref{l12} and $\mathcal B^{(n)}_M$, for $n=\pm1$, given by \eqref{cb1}.
Note that, in view of \eqref{LM}, we have
\[
\Z_{+,1} = (-\infty,M]\cap\Z, \quad \Z_{+,-1} = [-M,\infty)\cap\Z
\]
in this case. Denoting
\beq \label{FMpm}
F_M^-(\omega)=\tr \chi_{\mathcal B_M^{(1)}}f\pa{\Pi_{\Z_{+,1}}(\omega)}\chi_{\mathcal B_M^{(1)}},\quad F_{M}^+ (\omega)=\tr \chi_{\mathcal B_M^{(-1
)}}f\pa{\Pi_{\Z_{+,-1}}(\omega)}\chi_{\mathcal B_M^{(-1)}},
\eeq
we can rewrite \eqref{edec}  as
%
%
\beq \label{ffmfm}
\E\set{\abs{F_{\La_M}-\pa{F_{M}^+ +F_M^-}}^2}\le C e^{-cM}, \quad C <\infty,\quad  c>0.
\eeq
On the other hand, the expectation of all terms on the left of \eqref{ffmfm} are uniformly bounded in $M$  according to  \eqref{efmp}, hence
the terms are finite with probability 1. This and the Borel-Cantelli lemma yield the asymptotic relation
\beq\label{eq:FMpm}
F_{\La_{M}} (\omega)=F_{M}^+(\omega) +F_M^- (\omega)+ o(1),\quad  M \rightarrow\infty,
\eeq
which is valid with probability $1$.

Next, we have
\begin{align}\label{fmm}
&\hspace{-1cm}F_M^-(\omega)=
\sum_{x=M-[M/4]}^{M}f\pa{\Pi_{\Z_{+,1}}}(x,x,(\omega)) \\& =\sum_{x=-\infty}^Mf\pa{\Pi_{\Z_{+,1}}}(x,x,(\omega))-\Delta_M
= \tr f\pa{\Pi_{\Z_{+,1}}(\omega)}-\Delta_M (\omega), \notag
\end{align}
where
\[\Delta_M(\omega)=\sum_{x=-\infty}^{M-[M/4]-1}f\pa{\Pi_{\Z_{+,1}}}(x,x,\omega).\]
According to Lemma \ref{lem:keytech} (i),  $\E\set{\abs{\Delta_M}}\le Ce^{-cM}$ with $C<\infty$, $c>0$.  This and the Borel-Cantelli lemma  yield the relation
\beq\label{dm}
\Delta_M(\omega)=o(1),\quad M\rightarrow\infty,
\eeq
which is again valid with probability $1$.

Note now that according to \eqref{pic} and \eqref{Pi1},
$
\Pi_{\Z_{+,1}}(\omega)=\chi_{\Z_{+,1}}P(\omega)\chi_{(\Z_{+,1})^c}P(\omega)\chi_{\Z_{+,1}}.
$
This and \eqref{epr} yield
\[
\Pi_{\Z_{+,1}}(\omega)=\Pi_{\Z_+}(T_{M}\omega), \; \omega \in \Omega,
\]
 where $\Z_- $ is defined in \eqref{zpmd}. Combining \eqref{fmm} -- \eqref{dm} and their counterparts for the second term in \eqref{eq:FMpm}, we obtain
(\ref{Fpm}) --  \eqref{fp1}.
\end{proof}

\begin{theorem}\label{thm:varia}
Let $F_{\La}(\omega)$ be defined by \eqref{eq:FM}  and (\ref{picc}) -- (\ref{pic}),
where  $f$ satisfies Condition \ref{c:sob}, $P(\omega)$  satisfies Assumption \ref{assump}
and Assumption \ref{assump1}. Then,
for  $d\ge2$  and $\La_M=[-M,M]^d \; \subset \mathbb{Z}^d,
\; L=2M+1$, we have:
\begin{align}\label{eq:varmain'}
\Var\set{L^{-(d-1)}F_{\La_M}}:=\E\set{\pa{L^{-(d-1)}F_{\La_M}}^2}
-\pa{\E\set{L^{-(d-1)}F_{\La_M}}}^2\\ \le C\pa{\log M}^{4d/(d+1)}M^{-2(d-1)/(d+1)}.
\notag
\end{align}
\end{theorem}
\begin{proof}
The idea of the proof is to bound the variance  on the right of (\ref{eq:varmain'}) 
by that of the sum $\sum_{j=1}^m\eta_{j}$
of certain i.i.d. random variables $\{\eta_j \}_{j=1}^m$ and use the relation
\beq\label{eq:var2}
\Var\left\{\sum_{j=1}^m\eta_j\right\}=m\Var\{\eta_1\}.
\eeq
The result will then  follow from an appropriate choice for  $m=m(M)$.

To this end we will systematically use the bound
\beq\label{eq:var1}
\Var\{\xi_1\}\le 2\Var\{\xi_2\}+2\E\set{\abs{\xi_1-\xi_2}^2},
\eeq
which is valid for a pair of random variables $(\xi_1,\xi_2)$.

Let $\Z^d_{+,n}$ be a rigid lattice motion of $\Z^d_{+}$ such that $\La_M\subset \Z^d_{+,n}$ and $\mathcal F^{(n)}_M\subset \partial Z^d_{+,n}$. We will again use sets ${\mathcal B}_M^{(n)}$ defined in \eqref{eq:bufferset}, with $\ell$ as in \eqref{l12} and  $c>0$, which is  large enough but independent of $M$. Using \eqref{eq:var1} with $\xi_1= F_{\La_M}$ of \eqref{eq:FM} and
\[\xi_2=\sum_{n=-d}^{d}\tr \chi_{\mathcal B_M^{(n)}}f\pa{\Pi_{\Z^d_{+,n}}}\chi_{\mathcal B_M^{(n)}},\]
 we get that
 \beq\label{eq:pred}
\Var\{F_{\La_M}\}\le2\Var\left\{\xi_2\right\}+C (\log M)^{4 }M^{2(d-2)},
\eeq
 in view of \eqref{eq:finalspl} and \eqref{RM} with $d\ge 2$.
Thus, the inequality%
\begin{equation}
\mathbf{E}\Big\{\Big( \sum _{j}\eta _{j} \Big)\Big\}
^{2}\Big\}\leq \Big( \sum_{j}
\Big(\mathbf{E}\Big\{ \eta _{j}^{2}\Big\}\Big)^{1/2} \Big) ^{2},  \label{CS}
\end{equation}%
which implies
\begin{equation*}
\mathbf{Var}\Big\{\left( \sum _{j}\eta _{j}\right) ^{2}\Big\}\leq \Big( \sum
 _{j}\Big(\mathbf{E}\left\{ \eta _{j}^{2}\right\}\Big)^{1/2} \Big) ^{2},
\end{equation*}%
and properties \eqref{refp} -- \eqref{perp} yield%
\begin{equation}\label{vard}
\Var\{F_{M}\}\leq 4d^{2}\Var\left\{ \xi _{3}\right\} +C(\log M)^{4}M^{2(d-2)},
\end{equation}%
where%
\begin{equation}\label{xi3}
\xi _{3}=\tr\chi _{\mathcal{B}_{M}^{(1)}}f\left( \Pi _{\mathbb{Z}%
_{+,1}^{d}}\right) \chi _{\mathcal{B}_{M}^{(1)}}.
\end{equation}%
We introduce now the external surface layer $\mathcal{C}_{M}^{(1)}$, which
is the reflection of $\mathcal{B}_{M}^{(1)}$ with respect to the face $%
\mathcal{F}_{M}^{(1)}$ of \eqref{fac} without the points of the face, thus $%
\mathcal{B}_{M}^{(1)}\mathcal{\ }$belongs to our basic (closed) lattice cube
$\Lambda _{M}=[-M,M]^{d}$ and $\mathcal{C}_{M}^{(1)}$ belongs to its
exterior. We denote%
\begin{equation}
\mathcal{L}_{M}^{(1)}=\mathcal{B}_{M}^{(1)}\cup \mathcal{C}_{M}^{(1)}.
\end{equation}%
This lattice set is a rigid lattice motion  of  (cf.\ \eqref{eq:last})%
\begin{equation}\label{prism}
\pa{\lbrack -\ell ,\ell ]\times \lbrack -M+2\ell ,M-2\ell ]^{d-1}}\cap\
Z^d.
\end{equation}
Choose%
\begin{equation}
\ell =[c\log M],\;m\geq M^{1/d},  \label{elm}
\end{equation}%
where $c>0$ is large enough but independent of $M$ (see \eqref{l12}). Thinking of $\mathcal{L}_{M}^{(1)}$ as a set in $\mathbb{R}^{d}$ for just a moment, we partition it into $m^{d-1}$  congruent rectangular prisms $\{\widehat{%
\mathcal{L}}_{k}\}_{k=1}^{m^{d-1}}$, which are  rigid motions of (cf.\ \eqref{eq:last})
\begin{equation}
\lbrack -\ell ,\ell ]\times \lbrack -\tfrac{M-2\ell }{m},\tfrac{M-2\ell }{m}%
]^{d-1}\subset \mathbb{R}^{d}.  \label{calk}
\end{equation}%
Next, we introduce the lattice sets  \[\mathcal{L}_{k}=\widehat{\mathcal{L}}_{k}\cap \mathbb{Z}%
^{d},\;k=1,...,m^{d-1}.\] Adjusting the value of $m$ by $1$  (if necessary) we can make sure that these sets are congruent and
\begin{equation}\label{llie}
\mathcal{L}_{k^{\prime }}\cap \mathcal{L}_{k^{^{\prime \prime }}}=\mathcal{%
\varnothing },\;k^{\prime }\neq k^{\prime \prime }.
\end{equation}%
Let also%
\begin{equation}\label{lpp}
\mathcal{L}_{k}^{\prime }=\mathcal{B}_{M}^{(1)}\cap \mathcal{L}_{k},\;%
\mathcal{L}_{k}^{\prime \prime }=\mathcal{C}_{M}^{(1)}\cap \mathcal{L}_{k}
\end{equation}
be the parts of $\mathcal{L}_{k}$ belonging to our basic cube $\Lambda
_{M}=[-M,M]^{d}$ and its exterior respectively.

We will also need the lattice sets%
\begin{equation}\label{Kk}
\mathcal{K}_{k}=\left\{ x\in \mathcal{L}_{k}:\ \dist\left( x,\partial
\mathcal{L}_{k}\right) \geq \ell /4\right\} ,\;k=1,...,m^{d-1}
\end{equation}%
(the boundary $\partial \mathcal{C%
}$ for $\mathcal{C}\subset \mathbb{Z}^{d}$ is defined in \eqref{partialC}) and their parts%
\begin{equation}\label{Kkp}
\mathcal{K}_{k}^{\prime }=\mathcal{B}_{M}^{(1)}\cap \mathcal{K}_{k},\;%
\mathcal{K}_{k}^{\prime \prime }=\mathcal{C}_{M}^{(1)}\cap \mathcal{K}_{k},
\end{equation}%
belonging to $\Lambda _{M}=[-M,M]^{d}$ and its exterior respectively. Note that these sets are all separated by the "corridors" of width $\ell /2$ and that they are rigid motions of each other, for different values of $k$.

The volume of all the corridors
between the $\mathcal{K}_{k}^{\prime }$'s is given by
\begin{equation*}
\left\vert \mathcal{B}_{M}^{(1)}\setminus \cup _{k=1}^{m^{d-1}}\mathcal{K}%
_{k}^{\prime }\right\vert \leq Cm^{d-1}\ell ^{2}(M/m)^{d-2}=Cm\ell
^{2}M^{d-2}.
\end{equation*}%
Now, applying Lemma \ref{lem:keytech} (i) and \eqref{CS}, we get
\begin{align}
& \mathbf{E}\,\Big\{\Big|\tr\chi _{\mathcal{B}_{M}^{(1)}}f\left( \Pi _{%
\mathbb{Z}_{+,1}^{d}}\right) \chi _{\mathcal{B}_{M}^{(1)}}-\sum_{k}\tr\chi _{%
\mathcal{K}_{k}^{\prime }}f\left( \Pi _{\mathbb{Z}_{+,1}}\right) \chi _{%
\mathcal{K}_{k}^{\prime }}\Big|^{2}\Big\}  \label{eq:spl} \\
& \hspace{3cm} \leq Cm^{2}(\log M)^{4}M^{2(d-2)},  \notag
\end{align}%
in view of \eqref{elm}.
Hence, if we denote
\begin{equation*}
\xi _{4}=\sum_{k=1}^{m^{d-1}}\tr\chi _{\mathcal{K}_{k}^{\prime }}f\left( \Pi _{\mathbb{Z}%
_{+,1}^{d}}\right) \chi _{\mathcal{K}_{k}^{\prime }},
\end{equation*}%
and use \eqref{eq:var1}, \eqref{eq:pred} and \eqref{xi3}, we deduce:
\begin{equation}
\mathbf{E}\{|\xi _{3}-\xi _{4}|^{2}\}\leq 8d^{2}\Var\{\xi _{4}\}+Cm^{2}(\log
M)^{4}M^{2(d-2)}.  \label{eq:pred'}
\end{equation}%
Setting now
\begin{equation*}
\xi _{5}=\sum_{k=1}^{m^{d-1}}\tr f\left( \chi _{\mathcal{K}_{k}^{\prime }}\Pi _{\mathbb{Z%
}_{+,1}^{d}}\chi _{\mathcal{K}_{k}^{\prime }}\right) =\sum_{k=1}^{m^{d-1}}\tr f\left(
\Pi _{\mathcal{K}_{k}^{\prime },(\mathbb{Z}_{+,1}^{d})^{c}}\right)
\end{equation*}%
(see \eqref{picc} for the r.h.s. of the equality) and applying Lemma \ref{lem:compinf} to  $%
\mathcal{K}_{k}^{\prime }$ of \eqref{Kk} -- \eqref{Kkp} instead of $\mathcal{B}_{M}^{(n)}
$ of \eqref{eq:bufferset} and $M/m$ instead $M$, we obtain
\begin{equation}
\mathbf{E}\left\{ \left\vert \xi _{4}-\xi _{5}\right\vert ^{2}\right\} \leq
C\,m^{d-1}R_{d}(M/m)\leq C\,m^{2}(\log M)^{4}M^{2(d-2)}.  \label{eq:34pa}
\end{equation}%
Next, if (see \eqref{lpp} --  \eqref{Kkp})%
\begin{equation}\label{xi6}
\xi _{6}=\sum_{k=1}^{m^{d-1}}\tr f\left( \Pi _{\mathcal{K}_{k}^{\prime },\mathcal{L}%
_{k}^{\prime \prime }}\right) ,
\end{equation}
then we may apply \eqref{CS}, Lemma \ref{lem:trineq} and Lemma
\ref{l:picc} (iv) to get
\begin{align*}
\mathbf{E}\Big\{ \left\vert \xi _{5}-\xi _{6}\right\vert ^{2}\Big\}
&\leq \Big( \sum_{k=1}^{m^{d-1}}\Big(\mathbf{E}\Big\{ \Big\vert \tr f\Big( \Pi _{%
\mathcal{K}_{k}^{\prime },(\mathbb{Z}_{+,1}^{d})^{c}}\Big) -\tr f \Big(
\Pi _{\mathcal{K}_{k}^{\prime },\mathcal{L}_{k}^{\prime \prime }}\Big)
\Big\vert ^{2}\Big\}\Big)^{1/2} \Big) ^{2} \\
&\leq \Big( \sum_{k=1}^{m^{d-1}}\Big(\mathbf{E}\Big\{ \Big\Vert \Pi _{%
\mathcal{K}_{k}^{\prime },(\mathbb{Z}_{+,1}^{d})^{c}\setminus \mathcal{L}%
_{k}^{\prime \prime }} \Big\Vert _{\alpha }^{2\alpha }\Big\} \Big)^{1/2}
\Big) ^{2}. \notag
\end{align*}%
By construction, we have $\mathrm{dist}(\mathcal{K}%
_{k}^{\prime },(\mathbb{Z}_{+,1}^{d})^{c}\setminus \mathcal{L}_{k}^{\prime
\prime })\geq \ell $, so  according to Lemma \ref{lem:keytech}(ii) (with $\mathcal{C}_{1}=\mathcal{K}
_{k}^{\prime }$ and $\mathcal{C}_{2}=(\mathbb{Z}_{+,1}^{d})^{c}\setminus
\mathcal{L}_{k}^{\prime \prime }$) and \eqref{elm}, the r.h.s. is bounded by
\begin{equation*}
C^{1/2}|\partial \mathcal{K}_{k}^{\prime }|\mathrm{e}^{-\alpha \gamma
l/4}m^{d-1}\leq C_{1}((M/m)^{d-1}+\ell (M/m)^{d-2})\mathrm{e}^{-\alpha
\gamma l/4}=O(1),\;M\rightarrow \infty .
\end{equation*}%
We obtain that%
\begin{equation}\label{xi56}
\mathbf{E}\left\{ \left\vert \xi _{5}-\xi _{6}\right\vert ^{2}\right\} \leq
C_{2}M^{2(d-1)}\mathrm{e}^{-\alpha \gamma \ell /4}=O(1),\;M\rightarrow
\infty.
\end{equation}%
Finally, consider%
\begin{equation*}
\xi _{7}=\sum_{k}\tr f\left( \widehat{\Pi }_{\mathcal{K}_{k}^{\prime },%
\mathcal{L}_{k}^{\prime \prime }}\right) ,
\end{equation*}%
where (see \eqref{picc})
\begin{equation}
\widehat{\Pi }_{\mathcal{K}_{k}^{\prime },\;\mathcal{L}_{k}^{\prime \prime
}}=\chi _{\mathcal{K}_{k}^{\prime }}\widehat{P}_{\mathcal{L}_{k}}\chi _{%
\mathcal{L}_{k}^{\prime \prime }}\widehat{P}_{\mathcal{L}_{k}}\chi _{%
\mathcal{K}_{k}^{\prime }},  \label{pihat}
\end{equation}%
with $\widehat{P}_{\Lambda }$ defined in Assumption \ref{assump1},
i.e., $\widehat{\Pi }_{\mathcal{K}_{k}^{\prime },\mathcal{L}_{k}^{\prime
\prime }}$ is an analog of $\Pi _{\mathcal{K}_{k}^{\prime },\mathcal{L}%
_{k}^{\prime \prime }}$  with $P$ replaced by
$\widehat{P}_{\mathcal{L}_{k}}$.

Since $\dist\left( \mathcal{K}_{k}^{^{\prime }},\partial \mathcal{L}%
_{k}\right) \geq \ell /4$ (see \eqref{Kk} -- \eqref{Kkp}), we can apply Lemma \ref%
{lem:indistinga} with $a=\ell /4, \ \mathcal{Q}=\mathcal{L}_k$ and $\mathcal{Q}_a=\mathcal{K}_k$ to obtain

\begin{equation*}
\mathbf{E}\left\{ \left\vert \xi _{6}-\xi _{7}\right\vert ^{2}\right\} \leq
C\,\left\vert \partial \mathcal{L}_{k}\right\vert ^{5}\mathrm{e}^{-\alpha
\tilde{\gamma}\ell }\leq C\,\ell ^{5}m^{-5(d-1)}M^{5(d-1)}\mathrm{e}%
^{-\alpha \tilde{\gamma}\ell }.  \label{xi45}
\end{equation*}%
 In view of \eqref{elm}, we conclude that
\begin{equation}\label{xi67}
\mathbf{E}\left\{ \left\vert \xi _{6}-\xi _{7}\right\vert ^{2}\right\}
=O(1),\;M\rightarrow \infty .
\end{equation}%
Finally, we note that the random variables $\big\{\tr f\big(\widehat{\Pi }_{\mathcal{K}_{k}^{\prime },\mathcal{L}_{k}^{\prime \prime }}\big)\big\}%
_{k=1}^{m^{d-1}}$ are i.i.d. by Assumption \ref{assump1} and \eqref{llie},
so we can apply \eqref{eq:var2} for this collection of random variables. In addition, it follows from straightforward modifications to  Lemma \ref%
{corol:rough} that
\begin{equation*}
\Var\Big\{ \tr f\Big( \widehat{\Pi}_{\mathcal{K}_{k}^{\prime },\mathcal{L}_{k}^{\prime \prime }}\Big) \Big\}
\leq C\,\left\vert {\mathcal{K}}_{k}\right\vert ^{2}\leq
C_{1}\,(\log M)^{2}(M/m)^{2(d-1)}.
\end{equation*}%
This, bounds \eqref{vard}, \eqref{eq:pred'},  \eqref{eq:34pa},  \eqref{xi56},
and  \eqref{xi67} as well as the repeated use of \eqref{eq:var1} lead to the inequality
\begin{equation*}
\Var\{F_{\Lambda _{M}}\}\leq C(\log M)^{2}M^{2(d-1)}\left(
m^{1-d}+m^{2}(\log M)^{2}M^{-2}\right) .  \label{eq:pred3}
\end{equation*}%
Now the assertion of theorem follows by choosing $m=(\log
M)^{-2/(d+1)}M^{2/(d+1)}$ (cf.\ \eqref{elm}) and taking into account the
normalization factor $L^{-(d-1)}$ in \eqref{eq:varmain'}.
\end{proof}

\section{Auxiliary results}
We will start with several elementary assertions.

The  first one is a bound on $\E\abs{P(x,y)}^2$:
\beq\label{eq:Psqbn}
\E\{\abs{P(x,y)}^2\}\le \E\{\abs{P(x,y)}\}\le C\e^{-\gamma\abs{x-y}},
\eeq
where we used  \eqref{sr} and  Assumption \ref{assump}.
\begin{lemma}\label{l:h0} Let  $h:[0,1]\rightarrow \lbrack 0,1]$  be defined in  \eqref{h} and let
\beq\label{eq:h0}
h_0(t)=h\pa{\frac{1-\sqrt{1-4t}}{2}},\quad t\in[0,\tfrac{1}{4}].
\eeq
Then
\begin{enumerate}[(i)]
\item  $h_0$ is the function defined implicitly in
 \eqref{h0};
  \item $h_{0}(0)=0$;
\item $h_{0}$ is nonnegative, monotone increasing and concave on $[0,\tfrac{1}{4}]$;
\item  $4t\leq h_{0}(t)$  for $t\in[0,\tfrac{1}{4}]$;
\item For any  $\alpha\in(0,1)$, the function $h_{0}$ satisfies
\beq\label{eq:h_0der}
\max_{k\in\{0,1,2\}}\sup_{\in[0,1/4]}\abs{h_{0}^{(k)}(t)}\abs{t}^{k-\alpha}<\infty.
\eeq
\end{enumerate}
\end{lemma}
\begin{proof}
(i) and (ii) can be checked  directly.

(iii). It is straightforward to check (by taking two derivatives) that  $h$ is nonnegative and concave on $[0,1]$ and is monotone increasing on $[0,1/2]$. The assertion follows from the fact that $h_0$ is the re-parametrization of $h$, according to \eqref{eq:h0}.

(iv). Since $h_0(0)=0$ and $h_{0}(1/4)=1$,  the graph $y=h_0(t)$ and the line $y=4t$ intersect at $(0,0)$ and $(1/4,1)$. By  (i), $h_0$ is concave, which implies that the segment of the line $y=4t,\,t\in[0,1/4]$ lies below the graph of $h_0$.

(v). It follows from (\ref{h}) and (\ref{eq:h0}) that
\beq\label{eq:h0exp}
h_0(t)=1-\frac{1}{\log 2}\sum_{j=1}^{\infty }\frac{\pa{1-4t}^{j}}{2j(2j-1)},
\; t \in [0,1/4).
 \eeq
which means that $h_0$ is the analytic function in the disc $\set{z\in\C:\ \abs{z-1/4}<1/4}$. Hence, it suffices to consider the supremum over a smaller interval, say $[0,0.1]$ instead of $[0,1/4]$. We obtain, for  $t \in (0,0.1]$,
\beq
h_0(t)\le -4t\log_2t;\ \abs{h_0'(t)}\le \frac{2}{\log 2}-\frac{\log\pa{4t}}{\log2};\ \abs{h_0''(t)}\le \frac{1}{t\log 2}
\eeq
and the result follows.
\end{proof}

%
%
%
%
%
%
%

\medskip
We will also use the matrix valued Jensen inequality.
\begin{lemma}\label{l:pieq} Let $M=\{M_{jk} \}_{j,k=1}^n$ be an $n\times n$ hermitian matrix and $f:\mathbb{R}%
\rightarrow \mathbb{R}$ be a concave function. We have:
\begin{equation}
\pa{f(M)}_{jj}\leq f\pa{M_{jj}},\quad j=1,\ldots,n.  \label{pee}
\end{equation}
\end{lemma}
\begin{proof}
\bigskip According to the spectral theorem for hermitian matrices%
\[
\pa{f(M)}_{jj}=\int_{-\infty }^{\infty }f(\lambda )\mu _{j}(d\lambda ),
\]%
where $\mu _{j}$ is non-negative measure of total mass 1. Hence, by the
Jensen inequality and  the spectral theorem the r.h.s. is bounded from
above by
\begin{equation*}
f\pa{ \int_{-\infty }^{\infty }\lambda \mu _{j}(d\lambda )} =f\pa{M_{jj}} .
\end{equation*} \end{proof}
\begin{lemma}\label{l:picc}
Let $\mathcal C_1$ and $\mathcal C_2$ be non intersecting subsets of $\Z^d$ and $\Pi_{\mathcal C_1,\mathcal C_2}$ be defined by \eqref{picc}, where $P$ is an orthogonal projection in $\ell^2(\Z^d)$. We have
\begin{enumerate}
\item[(i)] $\Pi_{\mathcal C_1,\mathcal C_2}\ge 0$;
\item[(ii)] $\left\|\Pi_{\mathcal C_1,\mathcal C_2}\right\|\le 1$ and if $\mathcal C_1\subset \mathcal C_2^c:=\Z^d\setminus \mathcal C_2$, then $\left\|\Pi_{\mathcal C_1,\mathcal C_2}\right\|\le 1/4$;
\item[(iii)] If $\mathcal C'_1\subset \mathcal C''_1$, then $\Pi_{\mathcal C'_1,\mathcal C_2}$ is the restriction of $\Pi_{\mathcal C''_1,\mathcal C_2}$ to $\ell^2\pa{\mathcal C'_1}$;
\item[(iv)] If $\mathcal C_2=\cup_{j=1}^p\mathcal C_{2j}$ and $\set{\mathcal C_{2j}}_{j=1}^p$ are disjoint, then $\Pi_{\mathcal C_1,\mathcal C_2}=\sum_{j=1}^p\Pi_{\mathcal C_1,\mathcal C_{2j}}$.
\end{enumerate}
\end{lemma}
 Another simple yet useful observation is that if $A$ and $B$ is a pair of Hermitian operators on
 $\ell^2(\Z^d)$, and $f:\mathbb{R} \rightarrow \mathbb{C}$ is such that $f(0)=0$, then
\beq\label{eq:f(A+B)}
f(A)B=Bf(A)=0 \;\; \mbox{and} \;\; f(A+B)=f(A)+f(B), \mbox{ whenever } AB=BA=0.
\eeq
Both relations can be seen from the fact that there exists a basis $\{v_n\}$ on $\ell^2(\Z^d)$ that consists of eigenvectors for $A$ and $B$ with the property that either $Av_n=\lambda_n v_n\neq0$ and $Bv_n=0$, or  $Bv_n=\mu_n v_n\neq0$ and $Av_n=0$, or $Av_n=Bv_n=0$.

A natural tool for the trace estimates of an operator $A$ is its $\alpha$-Schatten norm
 \beq\label{tral}
\|A\|_\alpha=\pa{\|\abs{A}^\alpha\|_1}^{1/\alpha},
\eeq
where $\abs{A}=\pa{A^*A}^{1/2}$ and $\|\cdot\|_1$ is the trace norm, $\|A\|_1=\tr\abs{A}$.

Note that for $ \alpha\in (0,1)$ it is actually a quasi-norm, because it satisfies the modified triangle inequality:
$
\|A+B\|_\alpha \le C_{\alpha} \|A\|_\alpha+\|B\|_\alpha \quad C_\alpha=2^{1/\alpha-1}.
$
On the other hand,  the inequalities
\beq\label{eq:quasi}
\|A+B\|_\alpha^\alpha\le \|A\|_\alpha^\alpha+\|B\|_\alpha^\alpha,
\quad\alpha\in(0,1)
\eeq
and
\beq\label{nora}
\|AB \| _\alpha \le \|A\|\|B\|_\alpha; \quad \|AB \| _\alpha \le \|B\|\|A\|_\alpha
\eeq
hold in this case, and will be systematically used below.

\medskip
The next assertion plays an important role in our analysis.
\begin{lemma}\label{lem:trineq}
Assume that $A$ and $B$ are self-adjoint operators in the $\alpha$-Schatten class with $\alpha \in (0,1]$ such that $0\le A,B\le1/4$ and  that $f$ satisfies Condition \ref{c:sob}. Then $f(A)$ and $f(B)$ are trace class and
\beq\label{ftr}
\abs{\tr f(A)-\tr f(B)}\le C\left\|\abs{A-B}^\alpha\right\|_1 =C\left\|A-B\right\|_\alpha^\alpha.
\eeq
\end{lemma}
\begin{proof}
It follows from the definition of $\|\cdot\|_1$ that
\[\abs{\tr f(A)-\tr f(B)} = \abs{ \tr \pa{f(A) - f(B)}} \le \|f(A) - f(B)\|_1.\]
We will use now a particular case of Theorem 2.4 of \cite{So:15}, according to which if $f$ satisfies
Condition \ref{c:sob} and $\alpha \in (0,1]$, then for
\[\|f(A) - f(B)\|_1 \le C\left\|\abs{A-B}^\alpha\right\|_1,\]
provided $0\le A,B\le 1/4$.\footnote{In \cite{So:15}, $f$ has to be globally defined. The restriction $0\le A,B\le 1/4$ allows us to consider $f$ defined on the interval $[0,1/4]$.}
Combining the above two bounds, we obtain \eqref{ftr}. Plugging in \eqref{ftr} $B = 0$  ($A = 0$), we find that $f(A)$ (respectively $f(B)$) are trace class.
\end{proof}

\medskip
We will also need the following bound \beq\label{eq:sumsbnd}
\sum_{y\in \mathcal{C}^c}\e^{-\gamma\abs{x-y}}\le C\, \e^{-\tfrac{\gamma\dist(x,\partial \mathcal{C})}{2}}, \;\; x \in \mathcal{C}\subset\Z^d,
\eeq
where $\partial \mathcal{C}$ is the boundary of $\mathcal{C}$ (see \eqref{partialC}). Indeed, for $\mathcal{C}_x=\{y\in \mathcal{C}:\ \abs{x-y}\le \dist(x,\partial \mathcal{C})\}$, we have
\[\abs{x-y}\ge \dist(x,\partial \mathcal{C})+\dist(y, \mathcal{C}_x), \;\forall
y\in\mathcal{C}^c.\] This implies
\begin{align*}
\sum_{y\in\mathcal{C}^c}\e^{-\gamma\abs{x-y}}\le \e^{-\gamma\dist(x,\partial \mathcal{C})} \sum_{y\in\Z^d}\e^{-\gamma\dist(y, \mathcal{C}_x)} \le \e^{-\gamma\dist(x,\partial \mathcal{C})} \sum_{y\in\Z^d}\sum_{z\in \mathcal{C}_x}\e^{-\gamma\abs{y-z}}\\ =\e^{-\gamma\dist(x,\partial \mathcal{C})} \sum_{z\in \mathcal{C}_x}\sum_{y\in\Z^d}\e^{-\gamma\abs{y-z}}\le C\abs{\mathcal{C}_x} \e^{-\gamma\dist(x,\partial \mathcal{C})},
\end{align*}
and \eqref{eq:sumsbnd} follows from $\abs{\mathcal{C}_x}=\pa{2\dist(x,\partial \mathcal{C})+1}^d$.

A similar argument yields for any $\mathcal{C}\subset \Z^d$
\beq\label{eq:sumsbnd1}
\sum_{x\in\Z^d} \e^{- \gamma\dist\pa{x,  \mathcal{C}}} \le C\,\abs{\mathcal{C}}.
\eeq

We are now ready to establish important technical estimates of the paper:
\begin{lemma}\label{lem:keytech}
Let $\mathcal C_1$ and $\mathcal C_2$ be non-intersecting subsets of $\Z^d$. Suppose that Assumption \ref{assump} holds and $f$ satisfies Condition \ref{c:sob}. Then we have:
\begin{enumerate}
\item[(i)]
for $m= 1,2,...$
\beq\label{keytech}
\E\,\set{\abs{f\pa{\Pi_{\mathcal C_1,\mathcal C_2}}(x,x)}^m}\le  C_{m}\,\e^{-\tfrac{\alpha\gamma\dist(x,\partial \mathcal C_2)}{2}}, \; x \in \mathcal{C}_1,
\eeq
where $C_m$ depends only on $m$;
\item[(ii)]
\beq\label{cor:Schatnor}
\E\,\set{\left\|\Pi_{\mathcal C_1,\mathcal C_2}\right\|^{2\alpha}_\alpha}\le C\, \min\pa{\abs{\partial \mathcal C_1}^2,\abs{\partial \mathcal C_2}^2} \e^{-\tfrac{\alpha\gamma\dist\pa{\mathcal C_1,\mathcal C_2}}{2}};
\eeq
\item[(iii)] assume that $\mathcal C_1,\mathcal C_2$ are separated by an affine hyperplane $\mathcal H=\{x\in\Z^d:\ x\cdot e_n=k\}$ for some $n$ and $k$ and that the set $\mathcal C_1$ is contained in the infinite cylinder  $\mathcal D$ (that extends in $e_n$ direction), and that $\mathcal C_1\cap\mathcal H\neq\emptyset$. Let $\mathcal G$ be the cross-section of $\mathcal D$ with $\mathcal H$, i.e., $\mathcal G=\mathcal D\cap \mathcal H$. Then
\beq\label{lem:Shatnor1}
\E\,\set{\left\|\Pi_{\mathcal C_1,\mathcal C_2}\right\|_\alpha^{2\alpha}}\le C\, \abs{\mathcal G}^2\,\e^{-\tfrac{\alpha\gamma\dist\pa{\mathcal C_1,\mathcal C_2}}{2}}.
\eeq
\end{enumerate}
\end{lemma}
\begin{proof}
(i) It follows from \eqref{f} and the spectral theorem for hermitian matrices that
\[\abs{f\pa{\Pi_{\mathcal C_1,\mathcal C_2}}(x,x)}\le C\, (\Pi^\alpha_{\mathcal C_1,\mathcal C_2})(x,x).\]
We now apply Lemma \ref{l:pieq} to conclude that
\[\abs{f\pa{\Pi_{\mathcal C_1,\mathcal C_2}}(x,x)}\le C\, \pa{\Pi_{\mathcal C_1,\mathcal C_2}(x,x)}^\alpha.\]
This bound, \eqref{picc} and the inequality $0\le \Pi_{\mathcal C_1,\mathcal C_2}\le 1$ (see Lemma \ref{l:picc} (ii)) imply that for $m \ge 1$
\beq\label{ket}
\abs{f\pa{\Pi_{\mathcal C_1,\mathcal C_2}}(x,x)}^m\le C^m\, \pa{\Pi_{\mathcal C_1,\mathcal C_2}(x,x)}^{m\alpha}\le C^m\, \pa{\Pi_{\mathcal C_1,\mathcal C_2}(x,x)}^{\alpha}.
\eeq
Next, it follows from \eqref{picc} and \eqref{sr} that
\beq\label{fp}
\pa{\Pi_{\mathcal C_1,\mathcal C_2}(x,x)}^{\alpha}\le \pa{\sum_{y\in\mathcal C_2}\abs{P(x,y)}^2}^{\alpha}\le \pa{\sum_{y\in\mathcal C_2}\abs{P(x,y)}}^{\alpha}.
\eeq
Now, by using in the r.h.s. the H\"older inequality for expectations and \eqref{b_P}, we obtain
\begin{align}\label{eq:decf}
\E\set{\abs{f\pa{\Pi_{\mathcal C_1,\mathcal C_2}}(x,x)}^m}\le C^m\,\pa{\sum_{y\in\mathcal C_2}\E\set{\abs{P(x,y)}}}^{\alpha} \\ \notag \le C_{1m} \,\pa{\sum_{y\in\mathcal C_2}\e^{-\gamma\abs{x-y}}}^\alpha \le C_{m}\,\e^{-\tfrac{\alpha\gamma\dist(x,\partial \mathcal C_2)}{2}},
\end{align}
where  in the last step we used \eqref{eq:sumsbnd}.

(ii) Taking into account that $\Pi_{\mathcal C_1,\mathcal C_2}$, hence $\Pi^\alpha_{\mathcal C_1,\mathcal C_2}$, is positive definite, we deduce
\[\left\|\Pi_{\mathcal C_1,\mathcal C_2}\right\|_\alpha^{2\alpha}=\pa{\tr (\Pi_{\mathcal C_1,\mathcal C_2})^\alpha}^2=\pa{\sum_{x\in\mathcal C_1}(\Pi^\alpha_{\mathcal C_1,\mathcal C_2})(x,x)}^2.
\]
We now choose $f(x) = x^{\alpha}$ and  $m=1$ in Lemma \ref{lem:keytech} (i) and use inequalities \eqref{CS} and
\beq\label{ddd}
\dist (x, \partial \mathcal C_2) \ge \dist (x, \partial \mathcal C_1) + \dist (\mathcal C_1,\mathcal C_2),\quad x\in  \mathcal C_1
\eeq
to get
\begin{eqnarray*}
\mathbf{E}\{||\Pi _{\mathcal{C}_{1},\mathcal{C}_{2}}||_{\alpha }^{2\alpha
}\} &\leq &C\left( \sum_{x\in \mathcal{C}_{1}}\e^{-\frac{\alpha \gamma
\dist(x,\partial \mathcal{C}_{2})}{2}}\right) ^{2}\leq C\e^{-\alpha
\gamma  \dist(\mathcal{C}_{1},\mathcal{C}_{2})}\left( \sum_{x\in
\mathcal{C}_{1}}\e^{-\frac{\alpha \gamma \dist(x,\partial \mathcal{C}%
_{1})}{2}}\right) ^{2} \\
&\leq &C\e^{-\alpha \gamma \dist(\mathcal{C}_{1},\mathcal{C}%
_{2})}\left( \sum_{x\in \partial \mathcal{C}_{1}}\sum_{y\in \mathbb{Z}%
^{d}}e^{-\frac{\alpha \gamma |x-y|}{2}}\right) ^{2}\leq C|\partial \mathcal{C%
}_{1}|^{2}\e^{-\alpha \gamma \dist (\mathcal{C}_{1},\mathcal{C}_{2})}.
\end{eqnarray*}

But since in view of \eqref{picc}
\begin{equation}\label{Pires}
\Pi_{\mathcal C_1,\mathcal C_2}=AA^*, \; \;
\Pi_{\mathcal C_2,\mathcal C_1}=A^*A, \; \;
A=\chi_{\mathcal C_1}P\chi_{\mathcal C_2},
\end{equation}
we have $\left\|\Pi_{\mathcal C_1,\mathcal C_2}\right\|_\alpha=\left\|\Pi_{\mathcal C_2,\mathcal C_1}\right\|_\alpha$, and the result follows.

(iii) The assertion follows by the same argument used in Part (ii), if we replace  the bound \eqref{ddd} with
\[2\dist (x, \partial \mathcal C_2) \ge \dist (x, \partial \mathcal G) + \dist (\mathcal C_1,\mathcal C_2) ,\quad x\in  \mathcal C_1,\]
valid in this more restricted context.
\end{proof}
\begin{lemma}\label{corol:rough}
Suppose that Assumption \ref{assump} holds and $f$ satisfies Condition \ref{c:sob}. Then,  for any $d \ge 1$, we have
\[
\E\,\{\|f(\Pi_{\La_M})\|^2_1\}\le CM^{2(d-1)}
\]
and
\[
\E\,\{\|\chi_{\La_M}f(\Pi_{\Z^d_{+}})\chi_{\La_M}\|^2_1 \}\le CM^{2(d-1)},
\]
where $\La_M$, $\Pi_{\mathcal C}$, and $\Z^d_{+}$ are introduced  in \eqref{Lad}, \eqref{pic} and \eqref{zpmd}, respectively.
\end{lemma}
\begin{proof}
It follows from \eqref{f}, the spectral theorem and the definition \eqref{tral}  that
\[\|f(\Pi_{\La_M})\|_1\le C\left\|\Pi_{\La_M}\right\|^{\alpha}_\alpha.\]
Hence, Lemma \ref{lem:keytech} (ii) yields
\[
\E\set{\|f\pa{\Pi_{\La_M}}\|_1^2}\le C\abs{\partial\La_M}^2= C_1\,M^{2(d-1)}.
\]
On the other hand, using \eqref{CS} and Lemma \ref{lem:keytech} (i), we can bound
\begin{align*}
\E\,\Big\{\pa{\tr \chi_{\La_M}f(\Pi_{\Z^d_{+}})\chi_{\La_M}}^2\Big\}\le C\,
\Big(\sum_{x\in \La_M}\e^{-\tfrac{\alpha\gamma\dist(x,\partial  \Z^d_{+})}{2}}\Big)^2 \\ \le C\,\Big(\sum_{x\in\Z^d}\sum_{y\in \partial \Z^d_{+}\cap \La_M} \e^{-\tfrac{\alpha\gamma\abs{x-y}}{2}}\Big)^2 \le C_1\,\abs{\partial\La_M}^2= C_2\,M^{2(d-1)}.
\end{align*}
\end{proof}
\begin{lemma}\label{lem:buffer}
Suppose that Assumption \ref{assump} holds and $f$ satisfies Condition \ref{c:sob}. We have a bound
\beq\label{eq:hollow}
\E\,\Big\{\Big|\tr f\pa{\Pi_{\La_M}}-\sum_{n=-d}^{d}\tr \chi_{\mathcal B_M^{(n)}}f\pa{\Pi_{\La_M}}\chi_{\mathcal B_M^{(n)}}\Big|^2 \Big\}\le CR_d(M),
\eeq
where $\mathcal B_M^{(n)}$ and  $R_d(M)$ are defined  in \eqref{eq:bufferset} and  \eqref{RM} respectively.

\end{lemma}
\begin{proof}
Given $\ell \in \mathbb{N}$, write
\beq\label{eq:S_m}
\La_M=\La_{M-\ell}\cup \mathcal S_M\cup_{n=1}^{2d} \mathcal B_M^{(n)},\quad
\mathcal S_M=\La_M\setminus\pa{\La_{M-\ell}\cup_{n=1}^{2d} \mathcal B_M^{(n)}}
\eeq
to obtain
\begin{multline*}
\tr f\pa{\Pi_{\La_M}}-\sum_{n=1}^{2d}\tr \chi_{\mathcal B_M^{(n)}}f\pa{\Pi_{\La_M}}\chi_{\mathcal B_M^{(n)}} =\tr \chi_{\La_{M-\ell}}f\pa{\Pi_{\La_{M}}}\chi_{\La_{M-\ell}}+\tr \chi_{\mathcal S_M}f\pa{\Pi_{\La_M}}\chi_{\mathcal S_M}.
\end{multline*}
Since
\begin{align*}
&\hspace{-1cm}\abs{\mathcal S_M}=\abs{\La_M}-\abs{\La_{M-\ell}}-2d\abs{\mathcal  B_M^{(n)}}\\& = \pa{2M+1}^d-\pa{2M-2\ell+1}^d-2d\ell\pa{2M-4\ell-1}^{d-1}
\\& \hspace{1cm} \le C_d(d-1)\ell^2 M^{d-2},
\end{align*}
the application of \eqref{CS} and Lemma \ref{lem:keytech} (i) yield
\[\E\,\set{\abs{\tr \chi_{\La_{M-\ell}}f\pa{\Pi_{\La_{M}}}\chi_{\La_{M-\ell}}}^2}\le C\,\pa{M^d\,\e^{-\tfrac{\alpha\gamma\ell}{4}}}^2\]
and
\[\E\,\Big\{\abs{\tr \chi_{\mathcal S_M}f\pa{\Pi_{\La_M}}\chi_{\mathcal S_M}}^2\Big\}\le C\abs{\mathcal S_M}^2\le C\pa{(d-1)\ell^2 M^{d-2}}^2,\]
implying  the result.
\end{proof}
\begin{lemma}\label{lem:indisting}
Suppose that Assumption \ref{assump} holds and $f$ satisfies Condition \ref{c:sob}. Then, for each $n=\pm 1,...,\pm d$, we have
\beq\label{eq:indisting}
\E\set{\abs{\tr \chi_{\mathcal B_M^{(n)}}f\pa{\Pi_{\La_M}}\chi_{\mathcal B_M^{(n)}}-\tr f\pa{\chi_{\mathcal B_M^{(n)}}\Pi_{\Z^d_{+,n}}\chi_{\mathcal B_M^{(n)}}}}^2}  \le C\,R_d(M),
\eeq
where $R_d(M)$ is defined in \eqref{RM} and $\Z^d_{+,n}$ is a rigid lattice motion of $\Z^d_+$ such that $\La_M\subset\Z^d_{+,n}$ and $\mathcal{ F}_M^{(n)}\subset \partial \Z^d_{+,n}$.
\end{lemma}
\begin{proof}
 The proof is similar to that of Lemma \ref{lem:keytech}. Set \beq\label{tcm}
\mathcal{T}_M=\La_M\setminus\mathcal  B_M^{(n)},
\eeq
then
\[\Pi_{\La_M}-\chi_{\mathcal B_M^{(n)}}\Pi_{\La_M}\chi_{\mathcal B_M
^{(n)}}-\chi_{\mathcal T_M}\Pi_{\La_M}\chi_{\mathcal T_M}=\chi_{\mathcal T_M}\Pi_{\La_M}\chi_{\mathcal B_M^{(n)}}+h.c.\]
We claim that
\beq\label{eq:croster}
\E\,\Big\{\Big\| \chi_{\mathcal T_M}\Pi_{\La_M}\chi_{\mathcal B_M^{(n)}}\Big\|_\alpha^{2\alpha} \Big\} \le C\, R_d(M).
\eeq
This bound, \eqref{eq:f(A+B)}, Lemma \ref{lem:trineq} and \eqref{eq:quasi} yield
\beq\label{eq:compr1}
\E\set{\abs{\tr  \chi_{\mathcal B_M^{(n)}}f\pa{\Pi_{\La_M}}\chi_{\mathcal B_M^{(n)}}-\tr  f\pa{\chi_{\mathcal B_M^{(n)}}\Pi_{\La_M}\chi_{\mathcal B_M^{(n)}}}}^2}  \le C\,R_d(M).
\eeq
We also claim that
\beq\label{eq:indisting'}
\E\set{\left\| \chi_{\mathcal B_M^{(n)}}\Pi_{\La_M}\chi_{\mathcal B_M^{(n)}}- \chi_{\mathcal B_M^{(n)}}\Pi_{\Z^d_{+,n}}\chi_{\mathcal B_M^{(n)}}\right\|_\alpha^{2\alpha}} \le C\, R_d(M).
\eeq
Combining this bound with \eqref{eq:compr1}, \eqref{eq:f(A+B)}, Lemma \ref{lem:trineq},  and \eqref{eq:quasi},  we obtain the assertion of the lemma.
%
%
%
\vspace{.2cm}

To establish   \eqref{eq:indisting'} we note that by \eqref{picc} and 
Lemma \ref{l:picc} (iii)
\begin{align*}
&\left\| \chi_{\mathcal B_M^{(n)}}\pa{\Pi_{\La_M}-\Pi_{\Z^d_{+,n}}}\chi_{\mathcal B_M^{(n)}}\right\|_\alpha^{2\alpha}\\&=\left\| \Pi _{\mathcal{B}_{M}^{(n)},\Lambda _{M}^{c}}-\Pi _{\mathcal{B}_{M}^{(n)},(%
\mathbb{Z}_{+,n}^{d})^{c}}\right\|_{\alpha}^{2\alpha}=\left\|\Pi _{\mathcal{B}_{M}^{(n)},\Lambda
_{M}^{c}\setminus (\mathbb{Z}_{+,n}^{d})^{c}} \right\|_{\alpha}^{2\alpha},
\end{align*}
observe that  $\dist\pa{\mathcal B_M^{(n)},\La_M^c\setminus\pa{\Z^d_{+,n}}^c}\ge 2\ell$, and use Lemma \ref{lem:keytech} (ii).

\vspace{.2cm}

To get \eqref{eq:croster}, we further split $\mathcal T_M$ of \eqref{tcm}
as $\mathcal T_M=\mathcal R_M\cup \pa{\mathcal T_M\setminus \mathcal R_M}$, with \[\mathcal R_M=\{x\in\mathcal  T_M:\ \dist\pa{x, \pa{\Z^d_{+,n}}^c}\ge\ell\}.\]
Hence, using \eqref{eq:quasi}, we bound
\[\left\|\chi_{\mathcal T_M}\Pi_{\La_M}\chi_{\mathcal B_M^{(n)}}\right\|^{2\alpha}_\alpha\le C\,\pa{\left\|\chi_{\mathcal R_M}\Pi_{\La_M}\chi_{\mathcal B_M^{(n)}}\right\|^{2\alpha}_\alpha+\left\|\chi_{\mathcal T_M\setminus \mathcal R_M}\Pi_{\La_M}\chi_{\mathcal B_M^{(n)}}\right\|^{2\alpha}_\alpha}.\]
The second term on the right is estimated simply as
\begin{align*}
\left\|\chi_{\mathcal T_M \setminus \mathcal R_M }\Pi_{\La_M}\chi_{\mathcal B_M^{(n)}}\right\|^{2\alpha}_\alpha\le \left\|\chi_{\mathcal T_M\setminus \mathcal R_M}\right\|^{2\alpha}_\alpha\le \abs{\mathcal T_M\setminus \mathcal R_M}^{2}\\ \le C\pa{(d-1)\ell^2 M^{d-2}}^{2}\le C\, R_d(M).
\end{align*}
On the other hand,
we have\[ \chi_{\mathcal R_M}\Pi_{\La_M}\chi_{\mathcal B_M^{(n)}}=  \chi_{\mathcal R_M}P \,\chi_{\pa{\Z^d_{+,n}}^c}\,P\chi_{B_M^{(n)}}+\chi_{\mathcal R_M}\,P\, \chi_{\La_M^c\setminus \pa{\Z^d_{+,n}}^c}\,P\chi_{\mathcal B_M^{(n)}}.\]
This and \eqref{eq:quasi} imply
\begin{align*}
\left\|\chi_{\mathcal R_M}\Pi_{\La_M}\chi_{\mathcal B_M^{(n)}}\right\|^{2\alpha}_\alpha\le C\,\pa{\left\|\chi_{\mathcal R_M}P \chi_{\pa{\Z^d_{+,n}}^c}\right\|^{2\alpha}_\alpha+\left\| \chi_{\La_M^c\setminus \pa{\Z^d_{+,n}}^c}P\chi_{\mathcal B_M^{(n)}}\right\|^{2\alpha}_\alpha}\\ = C\,\pa{\left\|\Pi_{\mathcal R_M, \pa{\Z^d_{+,n}}^c}\right\|^{2\alpha}_{\alpha/2}+\left\| {\Pi_{ \La_M^c\setminus \pa{\Z^d_{+,n}}^c,\ \mathcal B_M^{(n)}}}\right\|^{2\alpha}_{\alpha/2}}.
\end{align*}
Recall now that \[\dist\pa{\mathcal R_M,\pa{\Z^d_{+,n}}^c}=\ell;\quad \dist\pa{\La_M^c\setminus \pa{\Z^d_{+,n}}^c,\mathcal B_M^{(n)}}=2\ell\] by construction.  This and Lemma \ref{lem:keytech} (ii) imply
 \[
 \E\,\pa{\left\|\Pi_{\mathcal R_M, \pa{\Z^d_{+,n}}^c}\right\|^{2\alpha}_{\alpha/2}+\left\| {\Pi_{ \La_M^c\setminus \pa{\Z^d_{+,n}}^c,\mathcal B_M^{(n)}}}\right\|^{2\alpha}_{\alpha/2}} \le C\,R_d(M).
 \]
 Putting all bounds together we get  \eqref{eq:croster}.
\end{proof}

\begin{lemma}\label{lem:compinf}
Suppose that Assumption \ref{assump} holds and $f$ satisfies Condition \ref{c:sob}. Then, for every $n=\pm 1,...,\pm d$, we have
\beq\label{eq:compinfd}
\E\set{\abs{\tr \chi_{\mathcal B_M^{(n)}}f\pa{\Pi_{\Z^d_{+,n}}}\chi_{\mathcal B_M^{(n)}}-\tr f\pa{\chi_{\mathcal B_M^{(n)}}\Pi_{\Z^d_{+,n}}\chi_{\mathcal B_M^{(n)}}}}^2}\le C\,R_d(M),
\eeq
where $\Z^d_{+,n}, \ \mathcal B_M^{(n)}$ and $R_d(M)$ are the same as in
the previous lemma.

\end{lemma}
\begin{proof}
We first observe that, thanks to \eqref{eq:f(A+B)},
\begin{align*}
& f\pa{\chi_{\mathcal B_M^{(n)}}\Pi_{\Z^d_{+,n}}\chi_{\mathcal B_M^{(n)}}} = \chi_{\mathcal B_M^{(n)}} f\pa{\chi_{\mathcal B_M^{(n)}}\Pi_{\Z^d_{+,n}}\chi_{\mathcal B_M^{(n)}}}\chi_{\mathcal B_M^{(n)}}\\& \hspace{1cm} = \chi_{\mathcal B_M^{(n)}} f\pa{\chi_{\mathcal B_M^{(n)}}\Pi_{\Z^d_{+,n}}\chi_{\mathcal B_M^{(n)}}+\chi_{\mathcal (B_M^{(n)})^c}\Pi_{\Z^d_{+,n}}\chi_{(\mathcal B_M^{(n)})^c}}\chi_{\mathcal B_M^{(n)}}.
\end{align*}
Hence the result  follows from Lemmas \ref{lem:trineq} and \eqref{eq:quasi} once we establish the bound
\beq\label{eq:interm}
\E\set{\left\|\Pi_{\Z^d_{+}}-\chi_{\mathcal B_M^{(n)}}\Pi_{\Z^d_{+}}\chi_{\mathcal B_M^{(n)}}-\chi_{\mathcal (B_M^{(n)})^c}\Pi_{\Z^d_{+}}{\chi_{\mathcal (B_M^{(n)})^c}}\right\|^{2\alpha}_\alpha}\le C\,R_d(M).
\eeq
To this end, using \eqref{eq:quasi}, we bound
\[
\E\Big\{\left\|\Pi_{\Z^d_{+}}-\chi_{\mathcal B_M^{(n)}}\Pi_{\Z^d_{+}}\chi_{\mathcal B_M^{(n)}}-\chi_{\mathcal (B_M^{(n)})^c}\Pi_{\Z^d_{+}}{\chi_{\mathcal (B_M^{(n)})^c}}\right\|^{2\alpha}_\alpha\Big\} \le C\,\E\Big\{\left\|\chi_{\mathcal B_M^{(n)}}\Pi_{\Z^d_{+}}\chi_{\mathcal (B_M^{(n)})^c}\right\|^{2\alpha}_\alpha\Big\}.
\]
The remainder of the argument strongly resembles
the one used to obtain %
\eqref{eq:indisting'}, with  usage of Lemma \ref{lem:keytech} (iii) instead of Lemma \ref{lem:keytech} (ii). We choose  $\mathcal{C}_{2}$ in the lemma to be the
semi-infinite parallelepiped $\mathcal{D}_{2}$ adjacent to the face $%
\widehat{\mathcal{F}}_{M}^{(n)}$ \eqref{fach} of $\mathcal{B}_{M}^{(n)}$ from the
exterior of $\Lambda _{M}$, i.e., belonging to $(\mathbb{Z}_{+}^{d})^{c}$. We choose  $\mathcal{C}_{1}$ to be the semi-infinite parallelepiped $%
\mathcal{D}_{1}$ adjacent to $\widehat{\mathcal{F}}_{M}^{(n)}$ and extended in the
direction opposite to that of $\mathcal{D}_{2}$ (so it belongs to $\mathbb{Z}%
_{+,n}^{d}$ and contains $\mathcal{B}_{M}^{(n)}$). For $%
\mathcal{G}$ in Lemma \ref{lem:keytech} (iii) we use $\widehat{\mathcal{F}}_{M}^{(n)}$.
Let $\mathcal{D}_{1}^{-}=\mathbb{Z}_{+,n}^{d}\setminus \mathcal{D}%
_{1}$ and let $\mathcal{D}_{2}^{-}=(\mathbb{Z}_{+,n}^{d})^{c}\setminus \mathcal{%
D}_2$. Then
\begin{align*}
& \chi _{\mathcal{B}_{M}^{(n)}}\Pi _{\mathbb{Z}_{+,n}^{d}}\chi
_{(B_{M}^{(n)})^{c}}=\chi _{\mathcal{B}_{M}^{(n)}}P\chi _{(\mathbb{Z}%
_{+,n}^{d})^{c}}P\chi _{\mathbb{Z}_{+}^{d}\setminus \mathcal{B}_{M}^{(n)}} \\
& \hspace{1cm}=\chi _{\mathcal{B}_{M}^{(n)}}P\chi _{\mathcal{D}_{2}}P\chi _{%
\mathbb{Z}_{+}^{d}\setminus \mathcal{B}_{M}^{(n)}}+\chi _{\mathcal{B}%
_{M}^{(n)}}P\chi _{\mathcal{D}_{2}^{-}}P\chi _{\mathbb{Z}_{+}^{d}\setminus
\mathcal{B}_{M}^{(n)}} , \notag
\end{align*}%
 in view of \eqref{picc},  so using \eqref{eq:quasi} and \eqref{nora} we get%
\begin{equation*}\label{eq:snag}
||\chi _{\mathcal{B}_{M}^{(n)}}\Pi _{\mathbb{Z}_{+,n}^{d}}\chi
_{(B_{M}^{(n)})^{c}}||_{\alpha }^{2\alpha }\leq C(||\chi _{\mathcal{D}%
_{2}}P\chi _{\mathbb{Z}_{+}^{d}\setminus \mathcal{B}_{M}^{(n)}}||_{\alpha
}^{2\alpha }+||\chi _{\mathcal{B}_{M}^{(n)}}P\chi _{\mathcal{D}%
_{2}^{-}}||_{\alpha }^{2\alpha }).
\end{equation*}%
To estimate  $\mathbf{E}\{ \|\chi _{\mathcal{D}_{2}}P\chi _{\mathbb{Z}%
_{+}^{d}\setminus \mathcal{B}_{M}^{(n)}}||_{\alpha }^{2\alpha }\}$   we first remove a set of cardinality $C(d-1)\ell ^{2}M^{d-2}$ from the
neighborhood of the "corners" of $\mathcal{D}_{2}$ common with $\mathbb{Z}%
_{+}^{d}\setminus \mathcal{B}_{M}^{(n)}$, to get $\widetilde{\mathcal{D}_{2}}$
such that $\dist(\widetilde{\mathcal{D}_{2}},\mathbb{Z}_{+}^{d}\setminus
\mathcal{B}_{M}^{(n)})\geq \ell $. Now we can bound $\E\,\{||\chi _{\widetilde{\mathcal{D}%
_{2}}}\Pi \chi _{{\mathbb{Z}_{+}^{d}\setminus \mathcal{B}_{M}^{(n)}}%
}||_{\alpha }^{2}\}$ using  Lemma \ref{lem:keytech} (iii), with  $
\mathcal{G}=\widehat{\mathcal{F}}_{M}$ and the estimate  $|\widehat{\mathcal{F}}_{M}|\leq (2M+1)^{d-1}$. This gives
\begin{align}\label{eq:bndB_M^c}
&\mathbf{E}\left\{ ||\chi _{\mathcal{D}_{2}}P\chi _{\mathbb{Z}%
_{+}^{d}\setminus \mathcal{B}_{M}^{(n)}}||_{\alpha }^{2\alpha }\right\} \\&
\hspace{0.5cm}\le\mathbf{E}%
\left\{ ||\chi _{\mathcal{D}_2 \setminus \widetilde{\mathcal{D}_{2}}}P\chi _{\mathbb{Z}%
_{+}^{d}\setminus \mathcal{B}_{M}^{(n)}}||_{\alpha }^{2\alpha }\right\} +\mathbf{E}%
\left\{ ||\chi _{\widetilde{\mathcal{D}_{2}}}P\chi _{\mathbb{Z}%
_{+}^{d}\setminus \mathcal{B}_{M}^{(n)}}||_{\alpha }^{2\alpha }\right\} \notag
\\& \hspace{1cm}\leq
C\left( |\mathcal{D}_{2}\setminus \widetilde{\mathcal{D}_{2}}|^{2}+\mathrm{e}^{-\frac{\alpha
\gamma \ell}{2}}|\partial
\widetilde{\mathcal{D}_2}|^2
\right) \leq CR_{d}(M), \notag
\end{align}%
in view of \eqref{eq:quasi} -- \eqref{nora}.
An analogous argument can be applied to the pair $(\mathcal{B}_{M}^{(n)},\mathcal{D}%
_{2}^{-})$ and yields
\begin{equation}\label{eq:bndB_M}
\mathbf{E}\left\{ ||\chi _{\mathcal{D}_{2}}P\chi _{\mathbb{Z}%
_{+}^{d}\setminus \mathcal{B}_{M}^{(n)}}||_{\alpha }^{2\alpha }\right\} \leq
CR_{d}(M).
\end{equation}%
Combining bounds \eqref{eq:bndB_M^c} and  \eqref{eq:bndB_M}, we arrive at %
\eqref{eq:interm}.
\end{proof}

\begin{lemma}
\label{lem:indistinga} Suppose that $P(\omega)$ is an ergodic projection satisfying
Assumption \ref{assump} and that $f$ satisfies Condition \ref{c:sob}. Let $%
\ell ,N$ be a pair of positive integers, let
\begin{equation*}
\mathcal{Q}=[-\ell ,\ell ]\times \lbrack -N,N]^{d-1}
\end{equation*}%
be a rectangular lattice prism, let
\begin{equation*}
\mathcal{Q}_{a}=\left\{ x\in \mathcal{Q}\cap \mathbb{Z}_{-}^{d}:\ \dist%
\left( x,\partial \mathcal{Q}\right) \geq a,\;a<\ell \right\} ;\;\mathcal{
Q}_{+}=Q\cap (\mathbb{Z}_{-}^{d})^{c}
\end{equation*}%
and let (cf.\ \eqref{pihat})
\begin{equation}
\widehat{\Pi }_{\mathcal{Q}_{a},\mathcal{Q}_{+}}(\omega)=\chi _{\mathcal{Q}_{a}}%
\widehat{P}_{\mathcal{Q}}\,(\omega)\chi _{\mathcal{Q}_{+}}\widehat{P}_{\mathcal{Q}%
}\,(\omega)\chi _{\mathcal{Q}_a},  \label{eq:mathcalP}
\end{equation}%
where $\widehat{P}_{\mathcal{Q}}(\omega)$ is defined in Assumption \ref{assump1}.

Then we have 
\begin{equation}
\mathbf{E}\Big\{ \left\vert \tr f\left( \widehat{\Pi }_{\mathcal{Q}_a,%
\mathcal{Q}_+}\right) -\tr f\left( \Pi _{\mathcal{Q}_{a},\mathcal{Q}%
_{+}}\right) \right\vert ^{2}\Big\} \leq C\,\left\vert \partial \mathcal{Q}%
\right\vert ^{5}\mathrm{e}^{-\alpha \,c\,a},  \label{eq:indistinga}
\end{equation}%
for some $C<\infty ,\ c>0$, and a sufficiently small value of $\ell /N$.
\end{lemma}


\begin{proof}
The assertion follows from the estimate
\begin{equation}  \label{eq:indistinga1}
\mathbf{E}\Big\{ \left\|\widehat{\Pi}_{\mathcal{Q}_a,\mathcal{Q}_+}-\Pi_{%
\mathcal{Q}_a,\mathcal{Q}_+} \right\|_\alpha^{2\alpha} \Big\} \le
C\,\left| \partial \mathcal{Q }\right|^5\mathrm{e}^{-c \, a}.
\end{equation}
Indeed, combining this bound with Lemma \ref{lem:trineq} and \eqref{eq:quasi},
 we get the desired result.

To obtain \eqref{eq:indistinga1} we first decompose
\begin{align*}
& \widehat{\Pi }_{\mathcal{Q}_{a},\mathcal{Q}_{+}}-\Pi _{\mathcal{Q}_{a},%
\mathcal{Q}_{+}} \\
& \hspace{1cm}=\chi _{\mathcal{Q}_{a}}\left( \widehat{P}_{\mathcal{Q}%
}-P\right) \chi _{Q_{+}}P\chi _{\mathcal{Q}_{a}}+\chi _{\mathcal{Q}_{a}}%
\widehat{P}_{\mathcal{Q}}\chi _{Q_{+}}\left( \widehat{P}_{\mathcal{Q}%
}-P\right) \chi _{\mathcal{Q}_{a}},  \notag
\end{align*}%
and use \eqref{eq:quasi} and \eqref{nora} to obtain
\begin{equation*}
\left\Vert \widehat{\Pi }_{\mathcal{Q}_{a},\mathcal{Q}_{+}}-\Pi _{\mathcal{Q}%
_{a},\mathcal{Q}_{+}}\right\Vert _{\alpha }^{2\alpha }\leq C\,\left\Vert
\chi _{\mathcal{Q}_{a}}\left( \widehat{P}_{\mathcal{Q}}-P\right) \chi _{%
\mathcal{Q}_{+}}\right\Vert _{\alpha }^{2\alpha }.
\end{equation*}%
Let $A=\chi _{\mathcal{Q}_{a}}\left( \widehat{P}_{\mathcal{Q}}-P\right) \chi
_{\mathcal{Q}_{+}}$, then $\Vert A\Vert \leq \Vert \widehat{P}_{\mathcal{Q}%
}-P\Vert \leq 2$, hence $0\leq A^{\ast }A/4\leq 1$, and so $\left\vert
A/2\right\vert ^{\alpha }=\left( A^{\ast }A/4\right) ^{\alpha /2}\leq
A^{\ast }A/4$. This leads to the bound
\begin{equation*}
\left\Vert \chi _{\mathcal{Q}_{a}}\left( \widehat{P}_{\mathcal{Q}}-P\right)
\chi _{\mathcal{Q}_{+}}\right\Vert _{\alpha }^{2\alpha }\leq C\,\left\Vert
\chi _{\mathcal{Q}_{+}}\left( \widehat{P}_{\mathcal{Q}}-P\right) \chi _{%
\mathcal{Q}_{a}}\left( \widehat{P}_{\mathcal{Q}}-P\right) \chi _{\mathcal{Q}%
_{+}}\right\Vert _{1}^{2},
\end{equation*}%
so \eqref{CS} yields
\begin{align*}
\mathbf{E}\Big\{ \left\Vert \widehat{\Pi }_{\mathcal{Q}_{a},\mathcal{Q}%
_{+}}-\Pi _{\mathcal{Q}_{a},\mathcal{Q}_{+}}\right\Vert _{\alpha }^{2\alpha
}\Big\} & \leq C\,\mathbf{E}\Big\{\Big(\sum_{x\in \mathcal{Q}%
_{+}}\sum_{y\in \mathcal{Q}_{a}}\left\vert \widehat{P}_{\mathcal{Q}%
}(x,y)-P(x,y)\right\vert ^{2}\Big)^{2}\Big\} \\
& \leq C\,\pa{{\sum_{x\in (\mathbb{Z}_{+}^{d})^{c}}\sum_{y\in \mathcal{Q}%
_{a}}\Big(\mathbf{ E}\Big\{\left\vert \widehat{P}_{\mathcal{Q}%
}(x,y)-P(x,y)\right\vert ^{4}\Big\}}\Big)^{1/2}}^{2}.
\end{align*}%
Since $\Vert \widehat{P}_{\mathcal{Q}}-P\Vert \leq 2$, we have $\left\vert
\widehat{P}_{\mathcal{Q}}(x,y)-P(x,y)\right\vert ^{4}\leq 8\left\vert
\widehat{P}_{\mathcal{Q}}(x,y)-P(x,y)\right\vert $. Using now Assumption \ref%
{assump1} and \eqref{CS}, we deduce
\begin{align*}
\mathbf{E}\Big\{ \left\Vert \widehat{\Pi }_{\mathcal{Q}_{a},\mathcal{Q}%
_{+}}-\Pi _{\mathcal{Q}_{a},\mathcal{Q}_{+}}\right\Vert _{\alpha }^{2\alpha
}\Big\} \leq C\,\pa{\sum_{x\in \mathcal{Q}_{+}}\sum_{y\in \mathcal{Q}_{a}}%
\Big(\mathbf{E}\left\{ \left\vert \widehat{P}_{\mathcal{Q}%
}(x,y)-P(x,y)\right\vert \right\}\Big)^{1/2} }^{2} \\
\leq C_{1}\,\left\vert \partial \mathcal{Q}\right\vert \Big(\sum_{x\in
\mathcal{Q}_{+}}\mathrm{e}^{-\gamma _{1}\dist\left( x,\partial \mathcal{Q}%
\right) }\sum_{y\in \mathcal{Q}_{a}}\mathrm{e}^{-\gamma _{1}\dist\left(
y,\partial \mathcal{Q}\right) }\Big)^{2}\leq C_{2}\,\left\vert \partial
\mathcal{Q}\right\vert ^{5}\mathrm{e}^{-\gamma _{1}a},
\end{align*}%
where in the last step we  used \eqref{eq:sumsbnd1} with $\partial
\mathcal{Q}$ as $\mathcal{C}$ and $\dist\left( \partial \mathcal{Q}%
_{a},\partial \mathcal{Q}\right) \ge a$.
\end{proof}

\medskip
We will now present  an important example of the discrete Schr\"{o}dinger
operator, for which Assumption \ref{assump1} is valid, and so Theorem \ref%
{thm:varia} and Result \ref{t:var} are applicable.

\begin{lemma}\label{l:iidpr}
Let $H(\omega)$ be the discrete Schr\"{o}dinger operator \eqref{DSe} acting in $l^{2}(%
\mathbb{Z}^{d})$ whose ergodic potential is%
\begin{equation}\label{VQ}
V(\omega)=gQ(\omega),\quad Q(\omega)=\left\{ Q(x,\omega)\right\} _{x\in \mathbb{Z}^{d}},
\end{equation}%
where $g>0$ and  $\left\{ Q(x,\omega)\right\} _{x\in \mathbb{Z}^{d}}$ is a
collection of i.i.d. random variables such that their common probability law
$F$ satisfies (\ref{hoc}). Let $H_\La(\omega)=H (\omega)|_{\La}$ be the (Dirichlet)
restriction
of $H$ to a domain $\La \subset \mathbb{Z}$
and let
\beq\label{pula}
P^{\La}(\omega)=\mathcal{E}_{H_{\La}(\omega)}(I)
\eeq
be the spectral projection of $H_{\La}(\omega)$ corresponding to an interval $I$.
Then a choice $\widehat{P}_{\La}(\omega)=P^{\La}(\omega)$ satisfies Assumption \ref{assump1}
for the cases described by items (a) -- (c) of the list below formula (\ref{b_P}).
\end{lemma}
\begin{proof}
It follows from the spectral theorem that if
\[G(\zeta):=(H-\zeta)^{-1}=\{G(\zeta;x,y)
\}_{x,y\in \mathbb{Z}^{d}},\ G^{\La}(\zeta):=(H_{\La}-\zeta)^{-1}
=\{G^{\La}(\zeta;x,y)\}_{x,y\in \La},\] and \[\widehat{ G}_{\La}(\zeta):=(H_{\La}\oplus
H_{\La^{c}}-\zeta)^{-1}=\{\widehat{G}_{\La}
(\zeta;x,y)\}_{x,y\in \mathbb{Z}^d}
\]
are the resolvents of $H,\ H_{\La}$ and $H_{\La}\oplus H_{\La^{c}}$,
then we have:
\begin{equation}
P(x,y)=\frac{1}{2\pi i}\oint_{K}G(\zeta ;x,y)d\zeta,   \; x,y \in \mathbb{Z}^d
\label{PHC}
\end{equation}%
and%
\begin{equation}
P^{\La}(x,y)=\frac{1}{2\pi i}\oint_{K}G^{\La}(\zeta
;x,y)d\zeta,  \; x,y \in \La \label{PHLC},
\end{equation}%
with probability $1$, \cite{Ai-Gr:98}. Here $K$ is a rectangular contour,
which encircles $I$ and crosses
transversally the real axis at the endpoints of $I$ (recall that the
probability that a given point of the real axis is an eigenvalue of $H$ is zero,
cf.\ Theorems 2.10, 2.12 and 4.21 in \cite{Pa-Fi:92}.

We will now use the resolvent identity
\begin{equation*}
G(\zeta)-\widehat{ G}_{\La}(\zeta)=G(\zeta)(H_{\La}\oplus H_{\La^{c}}-H)
\widehat{ G}_{\La}(\zeta),
\end{equation*}%
taking into account that the non-diagonal parts of $H$ and $H_{\La%
}\oplus H_{\La^{c}}$ are $-\Delta $ (see \eqref{DS}) and $-\Delta _{%
\La}\oplus \Delta _{\La^{c}}$, respectively. This gives
for  $x,y \in \Lambda$
\begin{align*}
&\hspace{-0.5cm} G(\zeta ;x,y)-G^{\La}(\zeta ;x,y)=G(\zeta ;x,y)-\widehat{ G}_{\La}(\zeta ;x,y)\\&=-\sum_{(u,v)\in \widetilde{%
\La}}G(\zeta ;x,u)\widehat{ G}_{\La}(\zeta ;v,y)=-\sum_{(u,v)\in \widetilde{%
\La}}G(\zeta ;x,u)G^{\Lambda }(\zeta ;v,y),
\end{align*}%
where $\widetilde{\La}$ is a collection of bonds between  points $%
v\in \partial \La$ and their nearest neighbors in $\La^{c}$.
This and the elementary inequality
\begin{equation*}
\Big(\sum |a_{j}|\Big)^{s}\leq \sum |a_{j}|^{s},\quad s\in \lbrack 0,1]
\end{equation*}%
give
\begin{equation*}
|G(\zeta ;x,y)-G^{\La}(\zeta ;x,y)|^{s/2}\leq \sum_{(u,v)\in
\widetilde{\La}}|(G(\zeta ;x,u)G^{\La}(\zeta
;v,y)|^{s/2},\;\;x,y\in \La.
\end{equation*}%
Taking expectation and using the H\"{o}lder
inequality in the r.h.s., we obtain
\begin{align}
& \mathbf{E}\{|G(\zeta ;x,y)-G^{\La}(\zeta ;x,y)|^{s/2}\}
\label{ress} \\
& \leq \sum_{(u,v)\in \widetilde{\La}}\Big(\mathbf{E}\{|G(\zeta
;x,u)|^{s}\}\,\mathbf{E}\{|G^{\La}(\zeta ;v,y)|^{s}\}\Big)^{1/2},\;\;x,y\in
\La.  \notag
\end{align}%
We will now use one of basic results of spectral theory of random
Schr\"{o}dinger operator, known as the fractional moment decay, see
\cite{Ai-Gr:98,Ai-Co:01,Ai-Wa:15}. It is the bound%
\begin{equation}
\mathbf{E}\left\{ |G(\lambda + i \varepsilon
;x,y)|^{s}\right\} \leq C\,(s)\mathrm{e}^{-\widetilde{\gamma}(s)\left\vert
x-y\right\vert },\; x,y\in \mathbb{Z}^{d},  \label{dynloc}
\end{equation}%
valid under the assumptions of the lemma, for some $s \in (0,\tau), \; C(s)<\infty,
\; \widetilde{\gamma}(s)>0$ and  uniformly in $\varepsilon \in
\mathbb{R}_{+}$ and  $\lambda\in I$. Similarly,  under the same
conditions we have
\begin{equation}
\mathbf{E}\{|G^{\La}(\lambda + i \varepsilon
;x,y)|^{s}\}\leq C\,(s)\mathrm{e}^{-\widetilde{\gamma}(s)\left\vert
x-y\right\vert }<\infty ,\;\lambda \in I,\;x,y\in \La.
\label{dynlocf}
\end{equation}%
These bounds, combined with the estimate
\begin{equation*}
| (A-\zeta )^{-1}(x,y)| \leq \Vert (A-\zeta )^{-1}\Vert \leq |\Im
\zeta |^{-1},
\end{equation*}%
valid for the resolvent of any self-adjoint operators, and  \eqref{PHC} --
\eqref{ress} imply for $x,y\in \La$%
\begin{align*}
& \mathbf{E}\{|P(x,y)-P^{\La}(x,y)|\} \\
& \hspace{1cm}\leq \frac{1}{2\pi }\oint_{K}\mathbf{E}\{|G(\zeta ;x,y)-G^{\La%
}(\zeta .x,y)|^{s/2}\}|\Im \zeta |^{s/2-1}|d\zeta | \\
& \hspace{1cm}\leq C_{0}|\partial \La|\mathrm{e}^{-\widetilde{\gamma }%
R}\oint_{K}|\Im \zeta |^{s/2-1}||d\zeta |\leq C|\partial \La|\mathrm{%
e}^{-\widetilde{\gamma }R}.
\end{align*}
Hence we  verified the condition \eqref{eq:rSch} of Assumption \ref{assump1} for $\widehat{P}_{\La}(\omega)=P^{\La}(\omega)$ in this context. The condition \eqref{xih} for such $\widehat{P}_{\La}(\omega)$ readily follows from the fact that for any $\La \in \mathbb{Z}^d$ the
projection $P^{\La}(\omega)$ depends only on the collection $\{V(x,\omega)\}_{x \in \La}$. Thus, if
$\La \cap \La_g=\emptyset$, then $P^{\La}(\omega)$ and $P^{\La_h}(\omega)$ are independent
and $\xi$ and $\xi_g$ in \eqref{xih} are i.i.d. since $\{V(x)\}_{x \in \mathbb{Z}^d}$
are i.i.d. variables.
   \end{proof}


\begin{remark}\label{r:frmom}
The condition of the lemma on the probability law of the i.i.d. potential
seem rather special, although they are  easy to check. For more general
but more involved conditions as well as for more general random operators
for which the basic bounds \eqref{dynloc} -- \eqref{dynlocf} hold see
\cite{Ai-Wa:15} and references therein. It is worth also mentioning that
\eqref{dynloc} -- \eqref{dynlocf} imply various other properties, which
are commonly associated with Anderson localization: spectral an dynamical
localization, exponential decay \eqref{b_P} of the projection kernel,
local Poisson statistic of eigenvalues, etc. (see  \cite{Ai-Wa:15,Ge-Ta:12,St:11}
for details and references).
\end{remark}

\section{Conclusions and Outlook}\label{concl}

Here, we discuss our results and some interpretations and implications thereof.
In this paper, we have proved  that entanglement entropy of free $d$-dimensional
disordered fermions satisfies the area law in the mean for any   $d \ge 1$  when
the Fermi energy lies in the exponentially localized part of the spectrum of the
one body Hamiltonian.
We  have also shown that for $d\ge 2$ fluctuations of
the entanglement entropy per unit surface area vanish as the block size
tends to infinity, i.e., that the entanglement entropy is selfaveraging for $d\ge 2$.

 The area law fails for translation invariant (clean) systems, which exhibit  so
 called logarithmic corrections to the  area
law (see e.g.\ \eqref{logd} and also
\cite{Am-Co:08,Ca-Co:11,Gi-Kl:06,He-Co:09,Le-Co:13,Le-Co:15,Pe-In:09,Wo:06}).
The difference in these two cases can be attributed to the inhibition of quantum
correlations (quantum coherence) of free fermions in the presence of  short-range
correlated spatial noise.  Mathematically,  the large block
behavior of  the entanglement entropy  for free fermions  is controlled by the
rate of decay at infinity  of the off-diagonal matrix elements of the Fermi
projection \eqref{P} (see also (\ref{PEI}))
of the associated one body Hamiltonian. In the clean case, this decay is slow
(see e.g.\ \eqref{pconv}),  since  the corresponding eigenstates  are
the plane waves. Consequently, in the clean case  free fermions display long
range quantum correlations and logarithmic corrections \eqref{logd} to the
area law. On the other hand, when the Fermi energy lies  in the exponentially
localized part of spectrum  of the one body  Hamiltonian, then the
eigenstates and the  Fermi  projection  decay exponentially\ (see  \eqref{b_P}.
This leads to short range quantum correlations and the area law  discussed
in this paper.
%

Let us remark that for one-dimensional many body systems  exponential decay of
all multipoint correlations (exponential clustering property) implies an upper
bound on area law, \cite{Br-Ho:15}. As was observed first in \cite{Ha:04,Ha:07},
exponential clustering occurs for ground states of systems where the ground state
energy is isolated from the rest of the many-body spectrum.
On the other hand,  in the disordered case (at least for the disordered free fermions),
the exponential decay results not from the existence of the spectral gap but from the
gapless spectrum with exponentially localized eigenfunctions
of one body Hamiltonian.

From this perspective, the absence  of  the logarithmic corrections to the
area law
in disordered fermions systems is   reminiscent of the absence of the d.c.
conductivity  and certain phase transitions (rounding effects) in macroscopic disordered systems, see e.g.\ \cite{Ai-Co:12,Ai-Wa:15,LGP}.

In this paper, we follow a bipartite implementation of the quantum system
(\ref{SBE}) -- (\ref{bip}), in which one first takes the "macroscopic" limit
$N \rightarrow \infty $ for the whole composite $\mathfrak{S}$ of (\ref{SBE})
and then studies the large block asymptotics
of the entanglement entropy \eqref{SB}, for $L\rightarrow \infty $.  It is
worth noting that one can consider another implementation of (\ref{SBE}) -- (\ref{bip}),
in which  $N$ and $L$ tend to infinity
simultaneously, say
\beq\label{meso}
L \sim N^{\alpha },\ \alpha \in (0,1],
\eeq
see e.g.\ \cite{Po-Co:15} for the case $\alpha=1$.

For disordered systems considered in this paper, the finite size effects  associated
with the "meso"-scaling (\ref{meso}) are negligible because of the exponential
localization, see (\ref{dynloc} -- \eqref{dynlocf}) and the   discussion in
\cite{Ai-Co:01,Ai-Wa:15,Ge-Ta:12}. Consequently, the  scaling  (\ref{meso})
of the whole bipartite composite and its block leaves our results essentially
unchanged.

On the other hand, the large block behavior of the entanglement entropy of
translation invariant  (clean) systems may be different in the scaling
(\ref{meso}). This
is already seen from the following simple argument. Set $\mathcal{D}=[-N,N]$
and $\Lambda =[-L_{1},L_{2}]\subset [-N,N]$ and take as one body Hamiltonian
the one dimensional
discrete Laplacian  $\Delta ^{\mathcal{D}}$. Then it follows from Lemma \ref{l:h0} (iv)
 that there exists $C>0$ such that
\[S_{\Lambda }\geq C|L_{2}-L_{1}|,\
L_{1,2}=c_{1,2} N^{\alpha _{1,2}},\ \alpha
_{1,2}\in (0,1].\]
Such behavior  is typical for the so called thermal entanglement
\cite{Am-Co:08,Ei-Co:11,Le-Co:15}, i.e.,\ for $S_{\Lambda }$ given by
\eqref{SB} with the Fermi projection $P$  replaced by a smooth function
of the Hamiltonian, \cite{Ki-Pa:15}, the
Fermi distribution $(1+\mathrm{e}^{\beta(t-\mu)})^{-1}, \; \beta < \infty$
in particular.
Note that $|L_1-L_2|$ is not the length (one dimensional volume) of our block
$[-L_1,L_2], $
 but rather a measure of its asymmetry inside our system $[-N,N]$.

 The results presented in Section 2 concern the von Neumann entanglement
 entropy \eqref{ente}, which for  free fermion systems can be written as
 the one body quantity \eqref{SB},  discussed in \eqref{SBD} -- \eqref{PL} and \eqref{h0} -- \eqref{sh0}.
However, our more general results of Sections 3 and 4 treat the more general
quantity \eqref{eq:FM}. Thus, we can apply the results of Section 3 and 4 to
the quantum analog of the R\'enyi  entropy (R\'enyi entanglement entropy),
\begin{equation}\label{Rege}
R^{(\alpha)}_\Lambda= (\alpha-1)^{-1}\log \mathrm{tr}_\Lambda \;
\rho_\Lambda^\alpha, \; \alpha \in (0, \infty), \; \alpha \neq 1,
\end{equation}
which, for free fermion systems, can be written as
\begin{equation}\label{Refe}
R^{(\alpha)}_\Lambda= \tr r_{\alpha} (P_\La), \; r^{(\alpha)} (t)
=(\alpha-1)^{-1} \log(t^\alpha + (1-t)^\alpha),
\end{equation}
i.e., again as a one body quantity.

We remark that  \eqref{Rege} converges to \eqref{ente} as   $\alpha \to 1$
provided that $\log_2$ in \eqref{ente} is replaced by $\log$, the base $e$ (natural)
logarithm. Similarly,  \eqref{Refe} converges to the expressions given by  \eqref{SB}
and \eqref{h} as $\alpha \to 1$ with the same replacement.
It is known,  \cite{Le-Co:13}, that the large block behavior of the R\'enyi
entanglement entropy of free translation invariant fermions is analogous to
that of the von Neumann entanglement entropy, i.e.,\ it has logarithmic
corrections to the  area law.

Since the function $r^{(\alpha)}_0$, defined by the relation
$
r^{(\alpha)}(t)=r^{(\alpha)}_0(t(1-t))
$
(cf.\ with its analogue $h_0$ in \eqref{h0}), satisfies Condition \ref{c:sob},
the results of Section 3 apply.
Thus, the large block behavior of the R\'enyi entanglement entropy of the disordered
free fermions is similar to that of the von Neumann entropy. In particular, the
expectation of the entanglement R\'enyi entropy satisfies the area law for any $d \ge
1$, the  entropy itself
is selfaveraging for $d \ge 2$, and   is not selfaveraging for $d=1$.


Finally, let us comment on a link of our results with the Szeg\H{o} theorem.
Recall that the theorem considers the large box asymptotics of $%
\mathrm{Tr\;}\varphi (A_{\Lambda })$, where  $A_\Lambda$ is the restriction
to $\Lambda \in \mathbb{Z}^d$ (see (\ref{Lad}) of a selfadjoint convolution
operator $A=\{A(x-y)\}_{x,y\in \mathbb{Z}%
^{d}}$  and $\varphi :%
\mathbb{R}\rightarrow \mathbb{C}$ is a function. It is known that
\cite{Bo-Si:90,Si:12}
\begin{equation}
\mathrm{Tr\,}\varphi (A_{\Lambda })=L^{d}\int_{\mathbb{T}^{d}}\varphi
(a(p))dp+
L^{d-1}\Psi _{d}+o(L^{d-1}),\;L\rightarrow \infty,
\label{szd}
\end{equation}%
where $a:\mathbb{T}^{d}\rightarrow \mathbb{R}$ is the Fourier transform of $%
\{A(x)\}_{x\in \mathbb{Z}^{d}}$ and $\Psi _{d}$ is a functional of $a$ and $%
\varphi $. The functions $a$ and $\varphi $ are known as the symbol of $A$ and
the test function. It is important to bear in mind that \eqref{szd} is valid
only if $a$ and
$\varphi $ are regular enough. However, if $a$ is piece-wise constant, then
the second (sub-leading) term of the formula is $\widetilde{\Psi }_{d}\
L^{d-1}\log L$. Not surprisingly, these asymptotic formulas are directly related
to the entanglement entropy and were in fact  strongly motivated
by quantum information theory \cite{Gi-Kl:06,He-Co:09,Le-Co:13,So:10,So:13,Wo:06}.

Let us confine ourselves to the case $d=1$ and consider the operator of
multiplication by $p$ in $L^{2}(\mathbb{T})$ as the symbol of the
self-adjoint operator $\widehat{p}$ in $\ell ^{2}(\mathbb{Z})$. Then we can
write the r.h.s. of \eqref{szd} as $\mathrm{Tr\,}\varphi ((a(\widehat{p}%
))_{\Lambda })$. This naturally leads us to consider a more general setting,
where one chooses a "standard" convolution operator $B$ and studies the
asymptotic behavior of $\mathrm{Tr\;}\varphi ((a(B))_{\Lambda })$ as $%
|\Lambda |:=L\rightarrow \infty $ and its dependence on the pair $(a,\varphi
)$. Note that convolution operators are  a particular case of ergodic operators
\eqref{eop}, namely, with $\Omega = \{0\}$. Thus, we can extend the above
general setting for the Szeg\H{o} theorem to ergodic operators by just choosing
a certain "standard" ergodic operator, say the discrete Schr\"{o}dinger
operator $H$ with ergodic potential \eqref{DS}, and study the asymptotics of
the random variable $\mathrm{Tr\,}\varphi ((a(H))_{\Lambda })$. A particular
case of this setting, where $a$ and $\varphi $ are smooth enough, was
considered in \cite{Ki-Pa:15}. In this case, whenever the potential is i.i.d.\
random, it was found that for $d=1$ the subleading term in the analog of
\eqref{szd} is not $\Psi _{1}$ but $L^{1/2}$ times a suitable Gaussian
random variable. Likewise, the case in which $\varphi =h$ (with $h$ defined in \eqref{h})
and $a=\chi _{(-\infty,\mu ]}$ corresponds to the entanglement entropy
\eqref{SB}, and Results \ref%
{t:expvan} -- \ref{t:var} establish
several new asymptotic forms of corresponding
traces in this (non-smooth) case.

\bigskip \noindent
\textbf{Acknowledgement} We wish to express our special thanks to J. Fillman for
careful reading of a draft of this manuscript, numerous corrections and suggestions
that markedly improved the final version. We are grateful to A. Sobolev for numerous
discussions and for drawing our attention to his work \cite{So:15}, which allowed
us to make the paper more transparent and the results (especially Result \ref{t:var}
and Theorem \ref{thm:varia}) stronger. L.P. would like to thank the Isaac Newton
Institute for Mathematical Sciences (Cambridge) for its hospitality during the
program "Periodic  and  Ergodic  Spectral Problems", May 2015 supported by EPSRC
Grant Number EP/K032208/1 and the Erwin Schr\"{o}dinger Institute (Vienna) for its
hospitality during the program "Quantum Many Body Systems, Random Matrices
and Disorder", July 2015. Financial support of grant 4/16-M of the National
Academy of Sciences of Ukraine is also acknowledged. 
A.E. is supported in part by NSF under grant DMS-1210982.

\end{document}